\documentclass[11pt]{article}

\usepackage[margin=1in]{geometry}
\usepackage{amsmath,amssymb,amsfonts,amsbsy,amsthm}
\usepackage{bm}
\usepackage{xcolor}
\usepackage{graphicx}
\usepackage{booktabs}
\usepackage{subcaption}
\graphicspath{{figures/}} 

\newcommand{\safeincludegraphics}[2][]{%
  \IfFileExists{#2}{\includegraphics[#1]{#2}}{%
    \IfFileExists{figures/#2}{\includegraphics[#1]{figures/#2}}{\fbox{Missing figure: \detokenize{#2}}}%
  }%
}
\usepackage{authblk}
\usepackage[numbers,sort&compress]{natbib}
\usepackage{hyperref}
\usepackage{mathtools}
\usepackage{bbm}
\usepackage{enumitem}

\theoremstyle{plain}
\newtheorem{lemma}{Lemma}[section]
\newtheorem{proposition}[lemma]{Proposition}
\newtheorem{corollary}[lemma]{Corollary}

\theoremstyle{definition}
\newtheorem{definition}[lemma]{Definition}

\theoremstyle{remark}
\newtheorem{remark}[lemma]{Remark}

\title{Disentangled Deep Priors for Bayesian Inverse Problems}
\author[1]{Arkaprabha Ganguli}
\author[1]{Emil Constantinescu}
\affil[1]{Mathematics \& Computer Science Division, Argonne National Laboratory, IL, USA\\
\text{\{aganguli, emconsta\}@anl.gov}}

\date{\today}

\begin{document}
\maketitle

\begin{abstract}
We propose a structured prior for high-dimensional Bayesian inverse problems based on a disentangled deep generative model whose latent space is partitioned into auxiliary variables aligned with known and interpretable physical parameters and residual variables capturing remaining unknown variability. This yields a hierarchical prior in which interpretable coordinates carry domain-relevant uncertainty while the residual coordinates retain the flexibility of deep generative models. By linearizing the generator, we characterize the induced prior covariance and derive conditions under which the posterior exhibits approximate block-diagonal structure in the latent variables---clarifying when representation-level disentanglement translates into a separation of uncertainty in the inverse problem. We formulate the resulting latent-space inverse problem and solve it using MAP estimation and Markov chain Monte Carlo (MCMC) sampling. On elliptic PDE inverse problems---conductivity identification and source identification — the approach matches an oracle Gaussian process prior under correct specification and provides substantial improvement under prior misspecification, while recovering interpretable physical parameters and producing spatially calibrated uncertainty estimates.
\end{abstract}

\section{Introduction}
\label{sec:introduction}

Bayesian inverse problems seek to infer unknown high-dimensional quantities $x$ from indirect and noisy measurements $y$ through a forward model
\begin{equation}
  y \;=\; \mathcal{F}(x) + \varepsilon,\qquad \varepsilon \sim p_{\varepsilon}.
  \label{eq:forward}
\end{equation}
A central difficulty is the construction of a prior $p(x)$ that is expressive enough to represent complex and potentially multimodal structure, while still supporting interpretable posterior summaries.

Deep generative models (DGMs), including variational autoencoders (VAEs), normalizing flows, and diffusion or flow-matching models, define implicit priors through learned maps $G_\theta$ from low-dimensional latent variables to high-dimensional fields. These models can represent complex data distributions, but their latent coordinates are often entangled: a single coordinate may influence several unrelated physical attributes. In inverse problems, this complicates posterior interpretation and limits the usefulness of latent-space summaries.

Recent work on auxiliary-variable-guided disentanglement introduces partial physical information during training. Aux-VAE partitions the latent space into a guided block aligned with known generative factors and a residual block that captures remaining variation, and it uses correlation-based regularizers to reduce leakage between these blocks~\cite{ganguli2025enhancing}. DL-CFM extends the same idea to latent conditional flow-matching models~\cite{ganguli2025dlcfm}.

The purpose of this study is to formulate these auxiliary-guided models as priors for Bayesian inverse problems and to analyze what part of the intended latent structure survives in the induced prior and posterior. To this end, we treat the guided variables $u$ as random variables with their own prior, retain a residual latent block $z_{\mathrm{rec}}$, and use the learned generator to induce a prior on $x$. This yields a hierarchical latent prior with user-specified coordinates that carry named semantics and with a residual block that represents variability not captured by $u$.

This leads to the notion of disentangled deep priors: implicit priors $p(x)$ realized by a deep generator $G_\theta$ and a structured latent prior $p(z_{\mathrm{aux}},z_{\mathrm{rec}},u)$ that encodes a factorized structure and alignment properties motivated by domain knowledge. The aim is to expose interpretable coordinates and to quantify residual coupling, rather than to assume exact posterior factorization.

\paragraph{Contributions} We make the following contributions:
\begin{enumerate}[leftmargin=*]
  \item We formalize a hierarchical Bayesian inverse problem in which the prior $p(x)$ is induced by a disentangled latent variable model trained with auxiliary-variable guidance. We show how the Aux-VAE and DL-CFM objectives can be reinterpreted as shaping the encoder-side conditional distribution toward a structured conditional prior  (via a Jensen/KL bound), and we clarify the additional generator-side conditions or diagnostics needed to obtain approximate decoupling properties for the induced (linearized) prior and posterior.
  \item We describe latent-space inference methods (variational and MCMC) that operate on interpretable latent coordinates $(u,z_{\mathrm{rec}})$ rather than directly on $x$, and we discuss how the resulting posterior can be summarized in terms of physical factors and residual structure together with diagnostics for residual coupling.
  \item We present numerical experiments on elliptic PDE inverse problems for spatially varying coefficients. In each case we describe how to define auxiliary variables $u$, how to train the disentangled generative model, and how to formulate the corresponding Bayesian inverse problem. The examples are used to assess reconstruction quality, posterior summaries, and calibration of the named coarse variables.
\end{enumerate}

\bigskip 

The remainder of this paper is organized as follows. Section~\ref{sec:background} reviews background material on Bayesian inverse problems and auxiliary-variable-guided disentanglement, and introduces the statistical interpretation used later in the paper. Section~\ref{sec:related} places the present formulation in the context of deep generative priors and disentangled latent-variable models. Section~\ref{sec:method} presents the hierarchical prior construction, the local analysis, and the latent-space inference formulations. Section~\ref{sec:experiments_2d} reports numerical experiments for elliptic PDE inverse problems. Section~\ref{sec:discussion} concludes with a discussion of limitations and possible extensions.

\section{Background}
\label{sec:background}

\subsection{Bayesian Inverse Problems}

Let $x \in \mathbb{R}^p$ denote an unknown high-dimensional parameter or field, and let $y \in \mathbb{R}^m$ denote data. Given a forward model $\mathcal{F} : \mathbb{R}^p \to \mathbb{R}^m$ and a noise model $p_\varepsilon$, the likelihood is $p(y \mid x) \propto p_\varepsilon\bigl(y - \mathcal{F}(x)\bigr)$. A prior $p(x)$ summarizes knowledge about $x$ before observing $y$. Bayes' rule then yields the posterior, $p(x \mid y) \propto p(y \mid x)\,p(x)$.

In many problems $x$ has a natural structure (fields, microstructures, profiles) and exhibits complex, non-Gaussian behavior. Classic Gaussian priors, e.g., based on Mat\'ern kernels or Gaussian Markov random fields, can encode smoothness or correlation length but are often inadequate to represent non-Gaussian features such as sharp interfaces, multimodality, or physically constrained patterns. Deep generative priors address this by representing $x$ as the output of a generator $G_\theta$ driven by latent variables $z$:
\begin{equation}
  z \sim p(z), \qquad x = G_\theta(z).
  \label{eq:deep-prior}
\end{equation}
The prior $p(x)$ is the push-forward of $p(z)$ through $G_\theta$, denoted $p(x) = (G_\theta)_\# p(z)$. DGMs such as VAEs and flows can approximate complex data distributions and have been used as priors in inverse problems, including PDE-based inverse problems and tomographic reconstruction. The latent variables $z$ are usually entangled, though, and lack physical interpretation; posteriors over $z$ are therefore difficult to analyze.

\subsection{Auxiliary-Variable-Guided Disentanglement}

We now summarize the Aux-VAE and DL-CFM frameworks that serve as building blocks for the disentangled priors.

\paragraph{Aux-VAE.} Let $x \in \mathbb{R}^p$ be data and $u \in \mathbb{R}^d$ be auxiliary variables that capture user-specified generative factors or descriptors (e.g., simulation inputs, or summary statistics extracted from $x$). A standard VAE introduces a latent $z \in \mathbb{R}^{d_Z}$ with an encoder $q_\phi(z \mid x)$, a decoder $p_\theta(x \mid z)$, and a prior $p(z)$. Aux-VAE partitions the latent:
\begin{equation*}
  z = (z_{\mathrm{aux}}, z_{\mathrm{rec}}),
  \quad
  z_{\mathrm{aux}} \in \mathbb{R}^d,
  \quad
  z_{\mathrm{rec}} \in \mathbb{R}^{d_Z - d},
\end{equation*}
and defines a conditional prior given auxiliary information $u$ of the form
\begin{equation}
  p(z \mid u)
  =
  \mathcal{N}\bigl(\mu_0(u),\Sigma_0\bigr),
  \qquad
  \mu_0(u) = \begin{bmatrix} u \\ 0 \end{bmatrix},\quad
  \Sigma_0 = \mathrm{diag}(\tau^2 I_d, I_{d_Z-d}),
  \label{eq:aux-prior}
\end{equation}
where $\tau^2 \ll 1$ is a small variance that softly anchors the guided coordinates to $u$ (after normalizing $u$).

The VAE objective for a single sample $(x,u)$ can be written as
\begin{equation}
  \mathcal{L}_{\mathrm{VAE}}(x,u;\theta,\phi)
  =
  \mathbb{E}_{q_\phi(z \mid x)}\bigl[\log p_\theta(x \mid z)\bigr]
  -
  \beta\,\mathrm{KL}\bigl(q_\phi(z \mid x)\,\|\,p(z \mid u)\bigr),
  \label{eq:vae-elbo}
\end{equation}
with a weight $\beta>0$.

Aux-VAE augments this objective with regularizers computed from the encoder mean $\mu_\phi(x) = \mathbb{E}[z\mid x]$. Writing $\mu_{\phi,\mathrm{aux}}$ and $\mu_{\phi,\mathrm{rec}}$ for the restrictions of $\mu_\phi$ to the auxiliary-guided and residual coordinates respectively, the additional terms enforce three properties: explicitness (each guided coordinate $\mu_{\phi,\mathrm{aux},j}$ tracks $u_j$ in a one-to-one fashion), intra-independence (different guided latents are decorrelated), and inter-independence (residual latents are decorrelated from $u$).

To achieve this, Aux-VAE introduces correlation-based penalties. Throughout, we view $(v,w)$ as random under the empirical training distribution (and the encoder mapping), and in practice we estimate all quantities below from minibatches. 

For random vectors $v \in \mathbb{R}^{m_v}$ and $w\in\mathbb{R}^{m_w}$ with finite second moments, let  $\Sigma_{vw} := \mathrm{Cov}(v,w)$, and $D_v := \mathrm{diag}(\Sigma_{vv})$, and thus $\mathrm{Corr}(v,w) := D_v^{-1/2}\,\Sigma_{vw}\,D_w^{-1/2}$.

\paragraph{Coordinatewise polynomial lift.}
For an integer $K\ge 1$ and a vector $v=(v_1,\dots,v_{m_v})$, we define the entrywise power 
$v^{(k)} := (v_1^k,\dots,v_{m_v}^k)$ for $k=1,\dots,K$ and the stacked feature map 
\begin{align}
\bigl(v^{(1)},v^{(2)},\dots,v^{(K)}\bigr) \in \mathbb{R}^{K m_v}\,.
\label{eq:poly_lift}
\end{align}
This lift contains coordinatewise monomials only.

Using these lifted features, we define two nonnegative dependence measures. The first is a decorrelation penalty that  averages absolute correlations between all pairs of lifted coordinates
\begin{equation}
  R_0^{(K)}(v,w)
  :=
  \frac{1}{K^2 m_v m_w}
  \sum_{k=1}^{K}\sum_{k'=1}^{K}
  \sum_{i=1}^{m_v}\sum_{j=1}^{m_w}
  \left|
    \bigl(\mathrm{Corr}(v^{(k)},w^{(k')})\bigr)_{ij}
  \right|\,.
  \label{eq:R0_def}
\end{equation}
Since $\bigl|\mathrm{Cov}(A,B)\bigr| =  \bigl|\mathrm{Corr}(A,B)\bigr|\,\sigma_A\sigma_B$, vanishing correlation implies vanishing covariance whenever the variances are nonzero; and thus, controlling correlations controls covariances provided the lifted standard deviations are bounded.

\paragraph{Alignment penalty and matching.}
When $m_v=m_w=m$ and the intended coordinate pairing is $i$-to-$i$, we define
\begin{equation}
  R_1^{(K)}(v,w)
  :=
  \frac{1}{K m}
  \sum_{k=1}^{K}
  \sum_{i=1}^{m}
  \left(
    1
    -
    \left|
      \bigl(\mathrm{Corr}(v^{(k)},w^{(k)})\bigr)_{ii}
    \right|
  \right).
  \label{eq:R1_def}
\end{equation}
In this paper $R_1^{(K)}$ is used either in the scalar case ($m=1$) or with identity matching.

Intuitively, $R_0^{(K)}$ penalizes cross-dependence between two blocks (via correlations of lifted features), while $R_1^{(K)}$ encourages one-to-one alignment between matched coordinates.
Aux-VAE then uses the following nonnegative penalties,
\begin{equation}
  \mathrm{Align}(A,B) = R_1^{(K)}(A,B),\quad
  \mathrm{Decorr}(A,B) = R_0^{(K)}(A,B),
\end{equation}
and constructs the total loss
\begin{align}
  \mathcal{L}_{\mathrm{AuxVAE}}
  &=
  -\,\mathcal{L}_{\mathrm{VAE}}
  + \lambda_1 \sum_{j=1}^d
  \Bigl[
    \mathrm{Align}(u_j,\mu_{\phi,\mathrm{aux},j})
    + \mathrm{Decorr}(u_j,\mu_{\phi,\mathrm{aux},-j})
  \Bigr]
  \nonumber\\
  &\quad
  + \lambda_2\,\mathrm{Decorr}(u,\mu_{\phi,\mathrm{rec}}),
  \label{eq:auxvae-loss}
\end{align}
with hyperparameters $\lambda_1,\lambda_2 \ge 0$.

We use a minor modification of the original Aux-VAE penalties~\cite{ganguli2025enhancing}: we express the objectives directly in terms of the correlation matrix of the lifted features \eqref{eq:poly_lift}, and (for the alignment term) we compare same-degree lifts so that perfect linear/sign alignment yields zero penalty. These choices give a clean interpretation in terms of coordinatewise polynomial correlations.

In the next proposition we recall a standard fact that under mild exponential-integrability conditions, independence is equivalent to the factorization of all mixed moments.
\begin{proposition}[Moment factorization and independence under exponential integrability]
\label{prop:moment_independence}
Let $X$ and $Y$ be real-valued random variables. Assume that the joint moment generating function
\[
  M_{X,Y}(s,t):=\mathbb{E}\bigl[e^{sX+tY}\bigr]
\]
is finite on a neighborhood of $(0,0)$; i.e., $\exists \delta>0$ such that $M_{X,Y}(s,t)<\infty$ $\forall$ $|s|<\delta$ and $|t|<\delta$.
Then the following are equivalent:
\begin{enumerate}
    \item $X$ and $Y$ are independent.
    \item All mixed moments factor: $\mathbb{E}[X^k Y^{k'}] = \mathbb{E}[X^k]\,\mathbb{E}[Y^{k'}]$ for all integers $k,k'\ge 0$.
    \item All centered mixed moments vanish: $\mathrm{Cov}(X^k, Y^{k'}) = 0$ for all integers $k,k'\ge 0$.
\end{enumerate}
Moreover, if (2) or (3) hold for all pairs $(k,k')$ with $k+k'\le K$, then $p_{X,Y}$ and $p_X p_Y$ agree on all mixed moments of total degree $\le K$.
\end{proposition}
\begin{proof}
The implications $(1)\Rightarrow(2)$ and $(2)\Rightarrow(3)$ are immediate. If $X$ and $Y$ are independent, then
\[
\mathbb{E}[f(X)g(Y)] = \mathbb{E}[f(X)]\,\mathbb{E}[g(Y)]
\]
for all measurable $f$ and $g$ for which the expectations exist. Taking $f(x)=x^k$ and $g(y)=y^{k'}$ gives $(2)$, and $(2)$ implies $(3)$ from
\[
\mathrm{Cov}(X^k,Y^{k'})=\mathbb{E}[X^kY^{k'}]-\mathbb{E}[X^k]\mathbb{E}[Y^{k'}].
\]

For $(3)\Rightarrow(1)$, fix $s,t\in\mathbb{R}$ with $|s|<\delta$ and $|t|<\delta$, and choose $\delta_0$ so that $\max\{|s|,|t|\}<\delta_0<\delta$. Since
\[
|sX+tY|\le \delta_0(|X|+|Y|),
\]
it is enough to show that $\mathbb{E}[e^{\delta_0(|X|+|Y|)}]<\infty$. For any real numbers $a,b$,
\[
e^{|a|+|b|}\le e^{a+b}+e^{a-b}+e^{-a+b}+e^{-a-b},
\]
so with $a=\delta_0 X$ and $b=\delta_0 Y$ we obtain
\[
 e^{\delta_0(|X|+|Y|)}
 \le
 e^{\delta_0 X+\delta_0 Y}
 +e^{\delta_0 X-\delta_0 Y}
 +e^{-\delta_0 X+\delta_0 Y}
 +e^{-\delta_0 X-\delta_0 Y}.
\]
Each expectation on the right is finite because the joint moment generating function is finite on a neighborhood of the origin. Hence Tonelli's theorem applies to the absolutely convergent double power series and
\[
  M_{X,Y}(s,t)
  =\sum_{k=0}^\infty\sum_{k'=0}^\infty \frac{s^k t^{k'}}{k!\,k'!}\,\mathbb{E}[X^k Y^{k'}].
\]
Under $(3)$, the mixed moments factor, so
\[
  M_{X,Y}(s,t)
  =\left(\sum_{k\ge 0}\frac{s^k}{k!}\,\mathbb{E}[X^k]\right)
   \left(\sum_{k'\ge 0}\frac{t^{k'}}{k'!}\,\mathbb{E}[Y^{k'}]\right)
  =M_X(s)\,M_Y(t)
\]
for all $|s|<\delta$ and $|t|<\delta$. By uniqueness of moment generating functions on a neighborhood of the origin, the joint law equals the product law, so $X$ and $Y$ are independent.

If the factorization holds only for $k+k'\le K$, then the mixed moments of total degree at most $K$ agree by definition.
\end{proof}

\begin{corollary}[Finite-order decorrelation implies finite-order moment factorization]
Let $X$ and $Y$ be real-valued random variables, and let $K \ge 1$. Assume that the moments below are finite and that
\begin{equation}
\mathrm{Cov}(X^k,Y^{k'}) = 0
\qquad
\forall\, 1 \le k,k' \le K.
\end{equation}
Then
\begin{equation}
\mathbb{E}[X^kY^{k'}] = \mathbb{E}[X^k]\,\mathbb{E}[Y^{k'}]
\qquad
\forall\, 1 \le k,k' \le K.
\end{equation}
In particular, for every pair $(r,s)$ with $r,s \ge 0$ and $r+s \le K$,
\[
\mathbb{E}[X^rY^s] = \mathbb{E}[X^r]\,\mathbb{E}[Y^s].
\]
Equivalently, $p_{X,Y}$ and $p_Xp_Y$ agree on all mixed moments of total degree at most $K$.
For fixed finite $K$, this does not in general imply that $X$ and $Y$ are independent.
\end{corollary}

\begin{proof}
For each $1 \le k,k' \le K$, the identity
\[
\mathrm{Cov}(X^k,Y^{k'})
=
\mathbb{E}[X^kY^{k'}] - \mathbb{E}[X^k]\mathbb{E}[Y^{k'}]
\]
implies that
\[
\mathbb{E}[X^kY^{k'}] = \mathbb{E}[X^k]\mathbb{E}[Y^{k'}]
\]
whenever $\mathrm{Cov}(X^k,Y^{k'})=0$. This proves the claimed factorization for all
$1 \le k,k' \le K$.

Now let $(r,s)$ satisfy $r,s \ge 0$ and $r+s \le K$. If $r=0$ or $s=0$, then
\[
\mathbb{E}[X^rY^s] = \mathbb{E}[X^r]\mathbb{E}[Y^s]
\]
holds trivially. Otherwise $1 \le r,s \le K$, so the first part applies and yields
\[
\mathbb{E}[X^rY^s] = \mathbb{E}[X^r]\mathbb{E}[Y^s].
\]
Hence $p_{X,Y}$ and $p_Xp_Y$ agree on all mixed moments of total degree at most $K$.
\end{proof}

\begin{remark}
The penalties $R_0^{(K)}$ and $R_1^{(K)}$ act on correlations between the lifted scalar features in
\eqref{eq:poly_lift}. In the limit, assuming the relevant variances are nonzero, $R_0^{(K)}(v,w)=0$ implies
\begin{equation}
\mathrm{Corr}(v_i^k,w_j^{k'}) = 0
\qquad
\forall\, i \in \{1,\ldots,m_v\},\ j \in \{1,\ldots,m_w\},\ 1 \le k,k' \le K,
\label{eq:coordwise_corr_zero}
\end{equation}
and therefore also $\mathrm{Cov}(v_i^k,w_j^{k'})=0$. Equivalently,
\begin{equation}
\mathbb{E}[v_i^k w_j^{k'}] = \mathbb{E}[v_i^k]\,\mathbb{E}[w_j^{k'}]
\label{eq:coordwise_moment_factor}
\end{equation}
for the same indices and degrees. Thus $R_0^{(K)}$ can be interpreted as encouraging coordinatewise polynomials to be uncorrelated.

Condition \eqref{eq:coordwise_corr_zero} is strictly weaker than independence of the random vectors $v$ and $w$. Vector independence would require factorization for all multivariate monomials (multi-indices), e.g.,
\begin{align}
\mathbb{E}[v_1^{k_1}\cdots v_{m_v}^{k_{m_v}} w_1^{\ell_1}\cdots w_{m_w}^{\ell_{m_w}}]
=
\mathbb{E}[v_1^{k_1}\cdots v_{m_v}^{k_{m_v}}]\,
\mathbb{E}[w_1^{\ell_1}\cdots w_{m_w}^{\ell_{m_w}}],
\qquad \forall (k,\ell).
\end{align}
The entrywise lift \eqref{eq:poly_lift} does not control mixed moments like $\mathbb{E}[v_1v_2w_j]$. Consequently, $R_0^{(K)}$ should be viewed as a practical proxy for reducing leakage between blocks, and it does not guarantee that $p(v,w)=p(v)p(w)$. If \eqref{eq:coordwise_moment_factor} holds for all degrees $k,k' \ge 0$ for a fixed pair $(i,j)$, then Proposition~\ref{prop:moment_independence}, applied to the scalar pair $(v_i,w_j)$, shows under its exponential-integrability assumption that $v_i$ and $w_j$ are independent. Even pairwise independence of every coordinate pair $(v_i,w_j)$ does not in general imply independence of the full vectors $(v,w)$.
\end{remark}
\begin{corollary}[Finite-sample convergence of the polynomial-correlation penalties]
\label{cor:finite_sample}
Let $(V,W)$ be a random pair with $V\in\mathbb{R}^{m_v}$ and $W\in\mathbb{R}^{m_w}$, and fix $K\ge 1$.
Assume:
\begin{enumerate}
  \item $\mathbb{E}[|V_i|^{4K}]<\infty$ $\forall i=1,\dots,m_v$ and $\mathbb{E}[|W_j|^{4K}]<\infty$ $\forall j=1,\dots,m_w$;
  \item $\mathrm{Var}(V_i^k) \ne 0$ and $\mathrm{Var}(W_j^{k'}) \ne 0$, $\forall 1\le i\le m_v$, $1\le j\le m_w$, and $1\le k,k'\le K$.
\end{enumerate}
Given $n$ iid\ samples $\{(V^{(\ell)},W^{(\ell)})\}_{\ell=1}^n$, let $\widehat{\mathrm{Corr}}(\cdot,\cdot)$ denote the sample correlation matrix obtained by replacing covariances/variances with their sample analogues.
Define $\widehat{R}_0^{(K)}$ and $\widehat{R}_1^{(K)}$ by \eqref{eq:R0_def}--\eqref{eq:R1_def} with $\mathrm{Corr}$ replaced by $\widehat{\mathrm{Corr}}$.
Then
\[
  \widehat{R}_0^{(K)}(V,W)\xrightarrow{\mathrm{a.s.}} R_0^{(K)}(V,W),
  \qquad
  \widehat{R}_1^{(K)}(V,W)\xrightarrow{\mathrm{a.s.}} R_1^{(K)}(V,W),
\]
and, moreover,
\[
  \widehat{R}_0^{(K)}(V,W)-R_0^{(K)}(V,W)=O_p(n^{-1/2}),
  \qquad
  \widehat{R}_1^{(K)}(V,W)-R_1^{(K)}(V,W)=O_p(n^{-1/2}).
\]
In particular, for fixed encoder parameters and iid samples, the minibatch penalties in \eqref{eq:auxvae-loss} consistently estimate their population counterparts when $(V,W)=(u,\mu_{\phi,\mathrm{rec}})$ or $(V,W)=(u,\mu_{\phi,\mathrm{aux}})$.
\end{corollary}
\begin{proof}
Each entry of a sample covariance matrix of lifted features is a sample mean of the form
\[
  \widehat{\mathrm{Cov}}(V_i^k,W_j^{k'})
  =
  \frac{1}{n}\sum_{\ell=1}^n 
  \Bigl(V_i^{(\ell)}\Bigr)^k \Bigl(W_j^{(\ell)}\Bigr)^{k'}
  -
  \left(\frac{1}{n}\sum_{\ell=1}^n \Bigl(V_i^{(\ell)}\Bigr)^k\right)
  \left(\frac{1}{n}\sum_{\ell=1}^n \Bigl(W_j^{(\ell)}\Bigr)^{k'}\right),
\]
and similarly for sample variances. Under assumptions (1)-(2), these sample moments obey a strong law of large numbers and a multivariate CLT yields $O_p(n^{-1/2})$ fluctuations.

For fixed $(i,j,k,k')$ and by assumption (2), the sample correlation coefficient can be written as the smooth map $(c,s_v,s_w)\ \longmapsto\ c/\sqrt{s_v s_w}$ applied to $(c,s_v,s_w)=(\widehat{\mathrm{Cov}}(V_i^k,W_j^{k'}),\widehat{\mathrm{Var}}(V_i^k),\widehat{\mathrm{Var}}(W_j^{k'}))$. The delta method then gives $\widehat{\mathrm{Corr}}(V_i^k,W_j^{k'})-\mathrm{Corr}(V_i^k,W_j^{k'})=O_p(n^{-1/2})$ for large enough $n$ almost surely. Finally, the absolute value is Lipschitz, so $\bigl||\hat R|-|R|\bigr|\le |\hat R-R|$ preserves the same rate, and finite averages over $(i,j,k,k')$ preserve both almost sure convergence and the $O_p(n^{-1/2})$ rate.
\end{proof}
The corollary above is a fixed-network iid statement: under standard moment assumptions, the minibatch dependence penalties converge to their population counterparts. During training the encoder parameters change across iterations, so this should be read as a local statistical justification for the minibatch estimators rather than as a full convergence statement for stochastic optimization.

\paragraph{DL-CFM.}
DL-CFM uses a VAE-style encoder with the same auxiliary-guided latent structure as Aux-VAE, coupled with a latent conditional flow matching generator~\cite{ganguli2025dlcfm}. A conditional flow matching model learns a time-dependent vector field $v_\theta(x_t,z,t)$ that transports samples from a simple reference distribution $p_0$ to the data distribution $p_1$ along a probability path $p_t$, by minimizing a regression loss of the form
\begin{equation}
  \mathcal{L}_{\mathrm{CFM}}
  =
  \mathbb{E}_{t,x_0,x_1,x_t}
  \bigl\|
    v_\theta(x_t,z,t) - u_t(x_t \mid x_0,x_1)
  \bigr\|_2^2,
\end{equation}
where $u_t$ is a prescribed conditional vector field and $x_t$ is drawn from the fixed probability path.

DL-CFM augments this with a KL term tying $q_\phi(z\mid x)$ to the conditional prior $p(z\mid u)$ and with the same alignment and decorrelation regularizers used in Aux-VAE, giving a loss 
\begin{align}
  \mathcal{L}_{\mathrm{DL\mbox{-}CFM}}
  &=
  \mathbb{E}_{t,z,x_0,x_1,x_t}
  \bigl\|
    v_\theta(x_t,z,t) - u_t(x_t \mid x_0,x_1)
  \bigr\|_2^2
  +
  \beta\,\mathbb{E}_{(x,u)}\,
  \mathrm{KL}\bigl(q_\phi(z\mid x)\,\|\,p(z\mid u)\bigr)
  \nonumber\\
  &\quad
  + \lambda_1 \sum_{j=1}^d
  \Bigl[
    \mathrm{Align}(u_j,\mu_{\phi,\mathrm{aux},j})
    + \mathrm{Decorr}(u_j,\mu_{\phi,\mathrm{aux},-j})
  \Bigr]
  +
  \lambda_2\,\mathrm{Decorr}(u,\mu_{\phi,\mathrm{rec}}).
  \label{eq:dlcfm-loss}
\end{align}
After training, the encoder provides an interpretable latent representation, while the CFM decoder produces conditional samples given $z$.

\section{Related Work}
\label{sec:related}

This study lies at the intersection of deep generative priors for inverse problems and auxiliary-guided or physics-aware disentanglement. We focus here on the relation between these two lines of work.

\paragraph{Deep Generative Priors for Inverse Problems.} Several studies represent an unknown field $x$ through a low-dimensional latent variable $z$ and a generator $G_\theta$, perform inference in latent space, and then map posterior samples back to the physical space~\eqref{eq:deep-prior}. Examples include GAN-based priors for physics-based inversion, transport-map interpretations of deep generative networks for Bayesian posteriors, autoencoding priors for MAP reconstruction, and hierarchical VAE priors used in plug-and-play regularization~\cite{patel2022deepprior,hou2019deepgenerativenetworks,gonzalez2022jointposterior,prost2023hvae_inverse_regularization}. These approaches can improve reconstruction quality and reduce the dimension of the inference problem, but the latent coordinates are typically not aligned with physical factors. As a result, posterior summaries in latent space are difficult to interpret.

\paragraph{Disentanglement and Physics-Aware Latent-Variable Models.} Aux-VAE and DL-CFM introduce auxiliary guidance so that a subset of latent coordinates aligns with named variables~\cite{ganguli2025enhancing,ganguli2025dlcfm,lipman2022flowmatching}. Related physics-aware VAE models introduce structural splits between known and unknown components and augment training with PDE residuals or adversarial terms~\cite{jacobsen2022disentangling,koune2025pivae}. These studies show that limited supervision and structured priors can produce interpretable latent coordinates, but they do not primarily address Bayesian inverse formulations with an explicit forward operator and likelihood.

In disentanglement more broadly, it is well understood that independence penalties alone do not identify factors without additional inductive bias~\cite{locatello2019challenging}. Relevant sources of inductive bias include structured priors, hyperpriors on latent statistics, and limited supervision~\cite{ansari2019hyperprior,jacobsen2022disentangling}. These observations motivate our emphasis on auxiliary variables $u$ and on separating encoder-side dependence control from generator-side diagnostics.

\paragraph{Present Work.} Our contribution is to combine these two threads in one Bayesian formulation. We use auxiliary-guided disentanglement to define a hierarchical prior with guided coordinates $u$ and residual coordinates $z_{\mathrm{rec}}$, and we then solve the induced inverse problem in latent space by optimization or sampling. Relative to existing deep generative priors, we introduce named latent coordinates with their own prior; relative to the disentanglement literature, we place that structure inside a Bayesian inverse problem with an explicit likelihood and forward map. The contribution is the formulation and the accompanying local analysis, rather than a new disentanglement architecture.

A complementary study constructs interpretable priors for Bayesian neural networks by selecting a parameter distribution that induces a prescribed Gaussian process covariance in function space via a Mercer expansion~\cite{alberts2026bayesian}. Unlike our approach, interpretability is attached to function-space covariance structure rather than to a disentangled latent decomposition into explicit physical factors and residual variability.

\section{Methodology: Disentangled Deep Priors}
\label{sec:method}

We use auxiliary-variable-guided disentanglement as an offline procedure to shape the latent space so that a subset of coordinates aligns with physically meaningful descriptors. The resulting generator induces a family of hierarchical priors in which prior information can be placed directly on interpretable variables, while the residual block represents variability that is difficult to model parametrically.

\subsection{Hierarchical Model}
\label{sec:hierarchical}
Let $x$ denote the unknown field or parameter of interest, and let $u \in \mathbb{R}^d$ be a user-chosen vector of auxiliary quantities associated with $x$. In the offline training stage we assume paired samples $(x,u)$ are available; $u$ may represent known simulation inputs or deterministic summaries $u=U(x)$ extracted from $x$, possibly only approximately. For example, in a PDE setting $u$ might contain coarse-scale statistics of a coefficient field. For notational simplicity we treat $u$ as already normalized in the same way as in training; physical units can be recovered by applying the inverse normalization map. The interpretation of the posterior over $u$ depends on this choice: if $u$ is a genuine generative input, then posterior uncertainty in $u$ has that physical meaning; if $u=U(x)$ is a summary, then the posterior over $u$ should be read as uncertainty over that summary induced by the prior and likelihood.

We assume that we have trained a disentangled generator $G_\theta$ with latent split
\begin{equation}
  z = (z_{\mathrm{aux}},z_{\mathrm{rec}}) \in \mathbb{R}^{d_Z},
\end{equation}
where $z_{\mathrm{aux}}\in\mathbb{R}^d$ is encouraged to align with the observed auxiliary variables $u$ (after the same normalization used in training), and $z_{\mathrm{rec}}\in\mathbb{R}^{d_Z-d}$ captures residual variability.

\paragraph{A soft-anchored prior.}
In the offline disentanglement stage, $u$ is observed and appears in the conditional prior $p(z\mid u)$ (see \S\ref{sec:background}). In the online inverse problem we instead treat $u$ as unknown and endow it with a prior $p(u)$ (often chosen as a simple distribution fit to the empirical $u$, e.g., $\mathcal{N}(0,I)$ in normalized coordinates). A hierarchical prior that keeps the roles of $u$ and $z_{\mathrm{aux}}$ distinct is:
\begin{subequations}
\begin{align}
  u &\sim p(u),
  \label{eq:prior-u}
  \\
  z_{\mathrm{aux}} \mid u &\sim \mathcal{N}(u,\tau^2 I_d),
  \label{eq:prior-zaux}
  \\
  z_{\mathrm{rec}} &\sim \mathcal{N}(0,I_{d_Z-d}),
  \label{eq:prior-zrec}
  \\
  x &= G_\theta(z_{\mathrm{aux}},z_{\mathrm{rec}}),
  \label{eq:prior-x}
  \\
  y \mid x &\sim p_\varepsilon\bigl(y-\mathcal{F}(x)\bigr).
  \label{eq:likelihood}
\end{align}
\end{subequations}
The hyperparameter $\tau^2>0$ controls how tightly the learned guided coordinates $z_{\mathrm{aux}}$ are connected to the physically meaningful $u$.

To make the implied factorization explicit, the latent prior in \eqref{eq:prior-u}--\eqref{eq:prior-zrec} can be written as
\begin{equation}
  p(u,z_{\mathrm{aux}},z_{\mathrm{rec}})
  = p(u)\,p(z_{\mathrm{aux}}\mid u)\,p(z_{\mathrm{rec}})
  = p(u)\,\mathcal{N}(z_{\mathrm{aux}};u,\tau^2 I_d)\,\mathcal{N}(z_{\mathrm{rec}};0,I_{d_Z-d}).
  \label{eq:prior-factorization}
\end{equation}
Moreover, \eqref{eq:prior-x} is deterministic, so one may equivalently write
\begin{equation}
  p(x\mid z_{\mathrm{aux}},z_{\mathrm{rec}})
  = \delta\bigl(x - G_\theta(z_{\mathrm{aux}},z_{\mathrm{rec}})\bigr),
  \label{eq:deterministic-decoder}
\end{equation}
where $\delta(\cdot)$ is a Dirac delta.
Substituting \eqref{eq:prior-x} into the likelihood \eqref{eq:likelihood} yields the latent-space likelihood
\begin{equation}
  p(y\mid z_{\mathrm{aux}},z_{\mathrm{rec}})
  = p_\varepsilon\!\left(y-\mathcal{F}\bigl(G_\theta(z_{\mathrm{aux}},z_{\mathrm{rec}})\bigr)\right)\,,
  \label{eq:likelihood-latent}
\end{equation}
and by applying Bayes' rule we obtain
\begin{equation}
  p(u,z_{\mathrm{aux}},z_{\mathrm{rec}}\mid y)
  \propto
  p(y\mid z_{\mathrm{aux}},z_{\mathrm{rec}})\,p(z_{\mathrm{aux}}\mid u)\,p(z_{\mathrm{rec}})\,p(u).
  \label{eq:posterior-factorization}
\end{equation}
Substituting \eqref{eq:likelihood-latent} and the Gaussian factors yields the explicit joint posterior
\begin{equation}
  p(u,z_{\mathrm{aux}},z_{\mathrm{rec}}\mid y)
  \propto 
  p_\varepsilon\!\left(y-\mathcal{F}\bigl(G_\theta(z_{\mathrm{aux}},z_{\mathrm{rec}})\bigr)\right)
  \,\mathcal{N}(z_{\mathrm{aux}};u,\tau^2 I_d)
  \,\mathcal{N}(z_{\mathrm{rec}};0,I_{d_Z-d})\,p(u).
  \label{eq:joint-posterior}
\end{equation}

\paragraph{A hard-constrained (collapsed) prior.}
In many applications the disentanglement training uses small $\tau^2$ and explicitness regularization, so that $z_{\mathrm{aux}}$ becomes (approximately) an identity representation of $u$. In this regime, a convenient and commonly used approximation is the hard-constrained limit $\tau\to 0$, in which $z_{\mathrm{aux}}=u$ deterministically and the generator can be written as $G_\theta(u,z_{\mathrm{rec}})$. The induced posterior is then
\begin{equation}
  p(u,z_{\mathrm{rec}} \mid y)
   \propto 
  p_\varepsilon \left(y - \mathcal{F}\bigl(G_\theta(u,z_{\mathrm{rec}})\bigr)\right)\,p(u)\,\mathcal{N}(z_{\mathrm{rec}};0,I_{d_Z-d}).
  \label{eq:latent-posterior}
\end{equation}
Either formulation gives a convenient coordinate system for posterior summaries: variation in $u$ tracks uncertainty in coarse interpretable factors, while variation in $z_{\mathrm{rec}}$ tracks residual fine-scale variability. This is a statement about how the latent coordinates are organized; it does not by itself imply posterior independence between the two blocks.

Although the induced prior on $x$ is implicit and generally non-Gaussian (as a push-forward through $G_\theta$), the factorization in \eqref{eq:prior-u}--\eqref{eq:prior-zrec} still separates the latent prior: uncertainty in the guided block is mediated through $u$ (equivalently through $z_{\mathrm{aux}}$ up to the soft constraint in \eqref{eq:prior-zaux}), while $z_{\mathrm{rec}}$ captures residual variability that is a priori independent of $u$. 

\begin{definition}[Generator tangent subspaces and overlap]
\label{def:rhoG}
Assume $G_\theta$ is differentiable at a latent point $z=(z_{\mathrm{aux}},z_{\mathrm{rec}})$. Define the tangent subspaces
\begin{align*}
&J_{\mathrm{aux}}(z) := \frac{\partial G_\theta}{\partial z_{\mathrm{aux}}}(z),
\qquad 
\mathcal{T}_{\mathrm{aux}}(z) := \operatorname{range}\!\big(J_{\mathrm{aux}}(z)\big),\\
&J_{\mathrm{rec}}(z) := \frac{\partial G_\theta}{\partial z_{\mathrm{rec}}}(z),
\qquad 
\mathcal{T}_{\mathrm{rec}}(z) := \operatorname{range}\!\big(J_{\mathrm{rec}}(z)\big).
\label{eq:tangent-subspaces}
\end{align*}
Let $Q_{\mathrm{aux}}(z)$ and $Q_{\mathrm{rec}}(z)$ be matrices with orthonormal columns spanning $\mathcal{T}_{\mathrm{aux}}(z)$ and $\mathcal{T}_{\mathrm{rec}}(z)$, respectively. Thus $\operatorname{range}(Q_{\mathrm{aux}}(z))=\mathcal{T}_{\mathrm{aux}}(z)$, $Q_{\mathrm{aux}}(z)^\top Q_{\mathrm{aux}}(z)=I$, $\operatorname{range}(Q_{\mathrm{rec}}(z))=\mathcal{T}_{\mathrm{rec}}(z)$, and $Q_{\mathrm{rec}}(z)^\top Q_{\mathrm{rec}}(z)=I$. We define the generator overlap metric
\begin{equation}
  \rho_G(z) \coloneqq \bigl\|Q_{\mathrm{aux}}(z)^\top Q_{\mathrm{rec}}(z)\bigr\|_2 \in [0,1],
  \label{eq:rhoG}
\end{equation}
which equals the cosine of the smallest principal angle between $\mathcal{T}_{\mathrm{aux}}(z)$ and $\mathcal{T}_{\mathrm{rec}}(z)$. In other words, a small $\rho_G(z)$ indicates that infinitesimal perturbations of $z_{\mathrm{aux}}$ and $z_{\mathrm{rec}}$ move $x=G_\theta(z)$ in nearly orthogonal directions.
\end{definition}

The next results are local. Under a first-order linearization of $G_\theta$, the induced covariance decomposes into guided and residual contributions. If the generator tangent subspaces have small overlap at relevant latent points, then the covariance of projected coordinates is approximately block diagonal. For inverse problems, a separate observation-space condition controls the corresponding block structure of the linearized posterior.

\paragraph{Linearized covariance structure and decoupling.} The prior in \eqref{eq:prior-u}--\eqref{eq:prior-x} is implicit: even though $(u,z_{\mathrm{aux}},z_{\mathrm{rec}})$ are endowed with simple Gaussian structure, the induced distribution on $x=G_\theta(z_{\mathrm{aux}},z_{\mathrm{rec}})$ is generally non-Gaussian. Through local linearizations of the generator and the forward operator, we make precise how the latent factorization yields an additive covariance decomposition at first order, and how approximate block structure can be quantified in terms of generator-subspace overlap and observation-space cross-sensitivities.
In the next lemma we show that, under a first-order linearization of the generator, the induced prior covariance decomposes additively into guided and residual contributions.
\begin{lemma}[Linearized induced prior covariance]
\label{lem:lin-cov}
Assume the soft-constrained prior \eqref{eq:prior-u}--\eqref{eq:prior-x} and suppose $u$ has finite second moment with mean $m_u \coloneqq \mathbb{E}[u]$ and covariance $\Sigma_u \coloneqq \mathrm{Cov}(u)$. Introduce the auxiliary noise variable
\begin{equation}
  \eta \sim \mathcal{N}(0,I_d),
  \qquad
  z_{\mathrm{aux}} \;=\; u + \tau\,\eta,
  \label{eq:zaux-as-u-plus-noise}
\end{equation}
so that \eqref{eq:prior-zaux} is satisfied and $(u,\eta,z_{\mathrm{rec}})$ are independent under the prior.

Consider the first-order (Taylor) approximation
\begin{equation}
  x_{\mathrm{lin}}
  \coloneqq
  G_\theta(m_u,0)
  + J_{\mathrm{aux}}\,(z_{\mathrm{aux}}-m_u)
  + J_{\mathrm{rec}}\,z_{\mathrm{rec}}\,
  \label{eq:x-linearization}
\end{equation}
where $J_{\mathrm{aux}}$ and $J_{\mathrm{rec}}$ are the Jacobians at $(z_{\mathrm{aux}},z_{\mathrm{rec}})=(m_u,0)$
\begin{equation}
  J_{\mathrm{aux}} \coloneqq \left.\frac{\partial G_\theta}{\partial z_{\mathrm{aux}}}\right|_{(m_u,0)} \in \mathbb{R}^{p\times d},
  \qquad
  J_{\mathrm{rec}} \coloneqq \left.\frac{\partial G_\theta}{\partial z_{\mathrm{rec}}}\right|_{(m_u,0)} \in \mathbb{R}^{p\times (d_Z-d)}.
  \label{eq:jacobians}
\end{equation}
Then $\mathbb{E}[x_{\mathrm{lin}}]=G_\theta(m_u,0)$ and the linearized prior covariance satisfies
\begin{equation}
  \mathrm{Cov}(x_{\mathrm{lin}})
  =\; J_{\mathrm{aux}}\,(\Sigma_u + \tau^2 I_d)\,J_{\mathrm{aux}}^\top
  \;+
  J_{\mathrm{rec}}\,J_{\mathrm{rec}}^\top.
  \label{eq:lin-prior-cov}
\end{equation}
In particular, the guided and residual contributions are uncorrelated at first order:
\begin{equation}
  \mathrm{Cov}\bigl(J_{\mathrm{aux}}(z_{\mathrm{aux}}-m_u),\,J_{\mathrm{rec}}z_{\mathrm{rec}}\bigr) = 0.
  \label{eq:lin-cross-cov-zero}
\end{equation}
\end{lemma}

\begin{proof}
From \eqref{eq:prior-zaux} we can write $z_{\mathrm{aux}}\mid u = u + \tau\eta$ with $\eta\sim\mathcal{N}(0,I_d)$, which yields \eqref{eq:zaux-as-u-plus-noise}.
By using \eqref{eq:x-linearization} and linearity of expectation we obtain %
\begin{align*}
  \mathbb{E}[x_{\mathrm{lin}}] = G_\theta(m_u,0)
     + J_{\mathrm{aux}}\,\mathbb{E}[z_{\mathrm{aux}}-m_u]
     + J_{\mathrm{rec}}\,\mathbb{E}[z_{\mathrm{rec}}].
\end{align*}
As $\mathbb{E}[z_{\mathrm{rec}}]=0$ by \eqref{eq:prior-zrec} and $\mathbb{E}[z_{\mathrm{aux}}] = \mathbb{E}[u + \tau\eta] = m_u$, we have $\mathbb{E}[z_{\mathrm{aux}}-m_u]=0$, and hence $\mathbb{E}[x_{\mathrm{lin}}]=G_\theta(m_u,0)$.

By definition
\begin{align*}
  \mathrm{Cov}(x_{\mathrm{lin}})
  &= \mathbb{E}\Big[(J_{\mathrm{aux}}(z_{\mathrm{aux}}-m_u) + J_{\mathrm{rec}}z_{\mathrm{rec}})
                  (J_{\mathrm{aux}}(z_{\mathrm{aux}}-m_u) + J_{\mathrm{rec}}z_{\mathrm{rec}})^\top\Big]\\
  &= J_{\mathrm{aux}}\,\mathbb{E}\big[(z_{\mathrm{aux}}-m_u)(z_{\mathrm{aux}}-m_u)^\top\big] J_{\mathrm{aux}}^\top
   + J_{\mathrm{aux}}\,\mathbb{E}\big[(z_{\mathrm{aux}}-m_u)z_{\mathrm{rec}}^\top\big] J_{\mathrm{rec}}^\top \\
  &\quad
   + J_{\mathrm{rec}}\,\mathbb{E}\big[z_{\mathrm{rec}}(z_{\mathrm{aux}}-m_u)^\top\big] J_{\mathrm{aux}}^\top
   + J_{\mathrm{rec}}\,\mathbb{E}\big[z_{\mathrm{rec}}z_{\mathrm{rec}}^\top\big] J_{\mathrm{rec}}^\top.
\end{align*}
For the first inner term, using \eqref{eq:zaux-as-u-plus-noise} and independence of $u$ and $\eta$, we have 
\begin{align*}
  \mathbb{E}\big[(z_{\mathrm{aux}}-m_u)(z_{\mathrm{aux}}-m_u)^\top\big] 
  &= \mathrm{Cov}(u + \tau\eta)
   = \mathrm{Cov}(u) + \tau^2\,\mathrm{Cov}(\eta) + \tau\,\mathrm{Cov}(u,\eta) + \tau\,\mathrm{Cov}(\eta,u) \\
  &= \Sigma_u + \tau^2 I_d + \tau\cdot 0 + \tau\cdot 0
   = \Sigma_u + \tau^2 I_d.
\end{align*}
Also by independence of $z_{\mathrm{rec}}$ from $(u,\eta)$, we have $\mathbb{E}\big[(z_{\mathrm{aux}}-m_u)z_{\mathrm{rec}}^\top\big] = 0$, $\mathbb{E}\big[z_{\mathrm{rec}}(z_{\mathrm{aux}}-m_u)^\top\big] = 0$, while by \eqref{eq:prior-zrec}, $\mathbb{E} [z_{\mathrm{rec}}z_{\mathrm{rec}}^\top]=I_{d_Z-d}$. 
Substituting these identities into $\mathrm{Cov}(x_{\mathrm{lin}})$ above leads to \eqref{eq:lin-prior-cov}. The cross-covariance  \eqref{eq:lin-cross-cov-zero} is the vanishing of the two mixed terms.
\end{proof}
At the linearized level, the guided and residual latent blocks therefore contribute independently to the prior covariance of the generated field.

\begin{remark}
Lemma~\ref{lem:lin-cov} is an exact covariance identity for the linear surrogate $x_{\mathrm{lin}}$ in \eqref{eq:x-linearization}. It relies on the factorization/independence in the latent prior and on a first-order Taylor approximation. In particular, it does not assert that the true nonlinear prior covariance $\mathrm{Cov}(x)$ is block diagonal, nor that higher-order terms in $G_\theta$ are negligible away from the linearization point.
\end{remark}

In the next corollary we bound the cross-covariance between guided and residual coordinates after projecting onto the corresponding generator tangent subspaces.
\begin{corollary}[Approximate block-diagonal covariance in a disentangled basis]
\label{cor:block-diag}
In the setting of Lemma~\ref{lem:lin-cov}, let $J_{\mathrm{aux}}=Q_{\mathrm{aux}}R_{\mathrm{aux}}$ and $J_{\mathrm{rec}}=Q_{\mathrm{rec}}R_{\mathrm{rec}}$ be QR decompositions, where $Q_{\mathrm{aux}}\in\mathbb{R}^{p\times r_{\mathrm{aux}}}$ and $Q_{\mathrm{rec}}\in\mathbb{R}^{p\times r_{\mathrm{rec}}}$ have orthonormal columns spanning the tangent subspaces induced by the guided and residual latents.

Define the projected coordinates
\begin{equation}
  x_{\mathrm{aux}} \coloneqq Q_{\mathrm{aux}}^\top (x_{\mathrm{lin}}-\mathbb{E}[x_{\mathrm{lin}}]),
  \qquad
  x_{\mathrm{rec}} \coloneqq Q_{\mathrm{rec}}^\top (x_{\mathrm{lin}}-\mathbb{E}[x_{\mathrm{lin}}]).
  \label{eq:projected-coordinates}
\end{equation}
Then the off-diagonal cross-covariance block obeys the bound
\begin{equation}
  \bigl\|\,\mathrm{Cov}(x_{\mathrm{aux}},x_{\mathrm{rec}})\,\bigr\|_2
  \;\le\;
  \underbrace{\|Q_{\mathrm{aux}}^\top Q_{\mathrm{rec}}\|_2}_{\eqqcolon\,\rho\in[0,1]}
  \Bigl(
    \|R_{\mathrm{aux}}(\Sigma_u+\tau^2 I_d)R_{\mathrm{aux}}^\top\|_2
    +
    \|R_{\mathrm{rec}}R_{\mathrm{rec}}^\top\|_2
  \Bigr).
  \label{eq:cross-cov-bound}
\end{equation}
In particular, if the subspaces $\mathcal{T}_{\mathrm{aux}}(z)$ and $\mathcal{T}_{\mathrm{rec}}(z)$ are (approximately) orthogonal so that $\rho\approx 0$, then the covariance of $(x_{\mathrm{aux}},x_{\mathrm{rec}})$ is (approximately) block diagonal.
\end{corollary}

\begin{proof}
The quantity $\rho=\|Q_{\mathrm{aux}}^\top Q_{\mathrm{rec}}\|_2$ is the overlap metric $\rho_G(z)$ (Definition~\ref{def:rhoG}) evaluated at the linearization point $z=(m_u,0)$.

From Lemma~\ref{lem:lin-cov} and the QR decompositions,
\begin{align*}
  \mathrm{Cov}(x_{\mathrm{lin}})
  &= J_{\mathrm{aux}}(\Sigma_u+\tau^2 I_d)J_{\mathrm{aux}}^\top + J_{\mathrm{rec}}J_{\mathrm{rec}}^\top \\
  &= Q_{\mathrm{aux}}\underbrace{\bigl(R_{\mathrm{aux}}(\Sigma_u+\tau^2 I_d)R_{\mathrm{aux}}^\top\bigr)}_{\eqqcolon\,S_{\mathrm{aux}}}Q_{\mathrm{aux}}^\top
   + Q_{\mathrm{rec}}\underbrace{\bigl(R_{\mathrm{rec}}R_{\mathrm{rec}}^\top\bigr)}_{\eqqcolon\,S_{\mathrm{rec}}}Q_{\mathrm{rec}}^\top.
\end{align*}
Now,
\begin{align*}
  \mathrm{Cov}(x_{\mathrm{aux}},x_{\mathrm{rec}})
  &= \mathrm{Cov}\bigl(Q_{\mathrm{aux}}^\top (x_{\mathrm{lin}}-\mathbb{E}[x_{\mathrm{lin}}]),\,Q_{\mathrm{rec}}^\top (x_{\mathrm{lin}}-\mathbb{E}[x_{\mathrm{lin}}])\bigr) \\
  &= Q_{\mathrm{aux}}^\top\bigl(Q_{\mathrm{aux}}S_{\mathrm{aux}}Q_{\mathrm{aux}}^\top + Q_{\mathrm{rec}}S_{\mathrm{rec}}Q_{\mathrm{rec}}^\top\bigr)Q_{\mathrm{rec}} \\
  &= S_{\mathrm{aux}}(Q_{\mathrm{aux}}^\top Q_{\mathrm{rec}}) + (Q_{\mathrm{aux}}^\top Q_{\mathrm{rec}})S_{\mathrm{rec}}.
\end{align*}
Taking the operator norm and applying submultiplicativity,
\begin{align*}
  \|\mathrm{Cov}(x_{\mathrm{aux}},x_{\mathrm{rec}})\|_2
  &\le \|S_{\mathrm{aux}}\|_2\,\|Q_{\mathrm{aux}}^\top Q_{\mathrm{rec}}\|_2 + \|Q_{\mathrm{aux}}^\top Q_{\mathrm{rec}}\|_2\,\|S_{\mathrm{rec}}\|_2 \\
  &= \|Q_{\mathrm{aux}}^\top Q_{\mathrm{rec}}\|_2\bigl(\|S_{\mathrm{aux}}\|_2 + \|S_{\mathrm{rec}}\|_2\bigr),
\end{align*}
which is \eqref{eq:cross-cov-bound} after substituting the definitions of $S_{\mathrm{aux}}$ and $S_{\mathrm{rec}}$.
\end{proof}
This corollary gives a concrete way to connect tangent-space overlap with block structure in the projected linearized coordinates. For the next development we consider the collapsed (hard-constraint) model \eqref{eq:latent-posterior} so that $x=G_\theta(u,z_{\mathrm{rec}})$. 
\begin{proposition}[Uniqueness and posterior decoupling under linearization]
\label{prop:identifiability}
Define the forward map in latent coordinates
\begin{equation}
  H(u,z_{\mathrm{rec}}) \coloneqq \mathcal{F}\bigl(G_\theta(u,z_{\mathrm{rec}})\bigr) \in \mathbb{R}^m.
  \label{eq:H-def}
\end{equation}
Assume $H$ is continuously differentiable at a reference point $(u_0,z_0)$ and let the Jacobian blocks be
\begin{equation}
  A_u \coloneqq \left.\frac{\partial H}{\partial u}\right|_{(u_0,z_0)} \in \mathbb{R}^{m\times d},
  \qquad
  A_{\mathrm{rec}} \coloneqq \left.\frac{\partial H}{\partial z_{\mathrm{rec}}}\right|_{(u_0,z_0)} \in \mathbb{R}^{m\times (d_Z-d)}.
  \label{eq:A-blocks}
\end{equation}
Consider the local linearization
\begin{equation}
  y \approx H(u_0,z_0) + A_u\,(u-u_0) + A_{\mathrm{rec}}\,(z_{\mathrm{rec}}-z_0) + \varepsilon,
  \qquad \varepsilon\sim\mathcal{N}(0,\Sigma_\varepsilon).
  \label{eq:lin-forward}
\end{equation}
Assume Gaussian priors $u\sim\mathcal{N}(m_u,\Sigma_u)$ with $\Sigma_u\succ 0$ and $z_{\mathrm{rec}}\sim\mathcal{N}(0,I)$, independent.

\begin{enumerate}
  \item \textbf{Unique posterior mean/MAP.} The negative log-posterior under the linear model \eqref{eq:lin-forward} is a strictly convex quadratic in $(u,z_{\mathrm{rec}})$, hence admits a unique minimizer. Equivalently, the linearized Gaussian posterior has a unique mean and coincident MAP.

  \item \textbf{Decoupling condition.} If
  \begin{equation}
    A_u^\top\Sigma_\varepsilon^{-1}A_{\mathrm{rec}} = 0,
    \label{eq:weighted-orthogonality}
  \end{equation}
  then the Gaussian posterior factorizes as
  $p(u,z_{\mathrm{rec}}\mid y) = p(u\mid y)\,p(z_{\mathrm{rec}}\mid y)$, i.e., the posterior covariance is block diagonal between the interpretable coordinates and the residual coordinates.

  \item \textbf{Noiseless local injectivity.} If the combined Jacobian $A\coloneqq [A_u\;A_{\mathrm{rec}}]\in\mathbb{R}^{m\times d_Z}$ has full column rank, then the linearized map $(u,z_{\mathrm{rec}})\mapsto H(u_0,z_0)+A_u(u-u_0)+A_{\mathrm{rec}}(z_{\mathrm{rec}}-z_0)$ is injective. Equivalently, in the noiseless linearized model $\varepsilon\to 0$, $(u,z_{\mathrm{rec}})$ are uniquely determined by $y$ (up to higher-order terms ignored by the linearization).
\end{enumerate}
\end{proposition}

\begin{proof}
Under \eqref{eq:lin-forward}, the linearized likelihood is Gaussian,
\begin{equation*}
  p(y\mid u,z_{\mathrm{rec}})
  \propto
  \exp\Bigl(-\tfrac12\bigl\|y - H(u_0,z_0) - A_u(u-u_0) - A_{\mathrm{rec}}(z_{\mathrm{rec}}-z_0)\bigr\|_{\Sigma_\varepsilon^{-1}}^2\Bigr).
\end{equation*}
Up to additive constants, the negative log-posterior is
\begin{align*}
  \Phi(u,z_{\mathrm{rec}})
  &= \tfrac12\bigl\|y - H(u_0,z_0) - A_u(u-u_0) - A_{\mathrm{rec}}(z_{\mathrm{rec}}-z_0)\bigr\|_{\Sigma_\varepsilon^{-1}}^2 \\
  &\quad + \tfrac12\|u-m_u\|_{\Sigma_u^{-1}}^2 + \tfrac12\|z_{\mathrm{rec}}\|_2^2\,,
\end{align*}
which is a quadratic function in $(u,z_{\mathrm{rec}})$.
The Hessian which corresponds to the Gaussian posterior precision is the block matrix
\begin{equation*}
  \Lambda
  =
  \begin{bmatrix}
    A_u^\top\Sigma_\varepsilon^{-1}A_u + \Sigma_u^{-1} & A_u^\top\Sigma_\varepsilon^{-1}A_{\mathrm{rec}}\\
    A_{\mathrm{rec}}^\top\Sigma_\varepsilon^{-1}A_u & A_{\mathrm{rec}}^\top\Sigma_\varepsilon^{-1}A_{\mathrm{rec}} + I
  \end{bmatrix}.
\end{equation*}

Let $A\coloneqq [A_u\;A_{\mathrm{rec}}]\in\mathbb{R}^{m\times d_Z}$ and let the (block-diagonal) prior covariance be
\begin{equation*}
  \Sigma_0 \coloneqq \mathrm{diag}(\Sigma_u, I) \in \mathbb{R}^{d_Z\times d_Z}.
\end{equation*}
Then the corresponding prior precision is $\Sigma_0^{-1}\succ 0$. Moreover, since $\Sigma_\varepsilon\succ 0$ we have $A^\top\Sigma_\varepsilon^{-1}A \succeq 0$.
The posterior precision can be written as 
\begin{equation*}
  \Lambda = \Sigma_0^{-1} + A^\top\Sigma_\varepsilon^{-1}A,,
\end{equation*}
which implies $\Lambda\succ 0$ and hence $\Phi$ is strictly convex and has a unique minimizer (Claim 1).

If \eqref{eq:weighted-orthogonality} holds, then $\Lambda$ is block diagonal. Its inverse, which is the posterior covariance matrix, is therefore also block diagonal. Since the posterior under \eqref{eq:lin-forward} is jointly Gaussian, block-diagonal covariance is equivalent to independence between $u$ and $z_{\mathrm{rec}}$. This proves Claim 2.

If $A=[A_u\;A_{\mathrm{rec}}]$ has full column rank, then $A\xi=0$ implies $\xi=0$. In particular, if
\begin{equation*}
  H(u_0,z_0) + A_u(u-u_0) + A_{\mathrm{rec}}(z_{\mathrm{rec}}-z_0)
  =
  H(u_0,z_0) + A_u(u'-u_0) + A_{\mathrm{rec}}(z'_{\mathrm{rec}}-z_0),
\end{equation*}
then $A_u(u-u') + A_{\mathrm{rec}}(z_{\mathrm{rec}}-z'_{\mathrm{rec}})=0$, i.e. $A\xi=0$ with $\xi=(u-u',\,z_{\mathrm{rec}}-z'_{\mathrm{rec}})$, and hence $\xi=0$ (Claim 3).

\end{proof}

\subsection{Prior Construction via Disentanglement}

The offline training objectives (Aux-VAE, DL-CFM) are designed so that the learned latent coordinates admit the same interpretation as the structured prior introduced in Section~\ref{sec:hierarchical}.

First, the KL term in~\eqref{eq:vae-elbo} (and its DL-CFM analogue) pushes the encoder distribution $q_\phi(z\mid x)$ toward the conditional prior $p(z\mid u)$ in~\eqref{eq:aux-prior}, which already encodes a soft constraint of $z_{\mathrm{aux}}$ to $u$ and an (approximately) factorized residual prior on $z_{\mathrm{rec}}$.

To connect this pointwise KL regularization to a statement that matches the training objective, let $p_{\mathrm{data}}(x,u)$ denote the empirical training distribution and define the aggregated (expected) variational posterior conditional on $u$ by
\begin{equation}
  q_\phi(z\mid u)
  \coloneqq
  \int q_\phi(z\mid x)\,p_{\mathrm{data}}(x\mid u)\,dx.
  \label{eq:agg-posterior}
\end{equation}
By convexity of $\mathrm{KL}(\cdot\|\cdot)$ in its first argument, for each fixed $u$ we have the Jensen-type bound
\begin{equation}
  \mathrm{KL}\bigl(q_\phi(z\mid u)\,\|\,p(z\mid u)\bigr)
  \le
  \mathbb{E}_{x\sim p_{\mathrm{data}}(\cdot\mid u)}\Big[\mathrm{KL}\bigl(q_\phi(z\mid x)\,\|\,p(z\mid u)\bigr)\Big].
  \label{eq:kl-jensen}
\end{equation}
Finally, averaging over $u\sim p_{\mathrm{data}}$ yields
\begin{equation}
  \mathbb{E}_{u\sim p_{\mathrm{data}}}\Big[\mathrm{KL}\bigl(q_\phi(z\mid u)\,\|\,p(z\mid u)\bigr)\Big]
  \le
  \mathbb{E}_{(x,u)\sim p_{\mathrm{data}}}\Big[\mathrm{KL}\bigl(q_\phi(z\mid x)\,\|\,p(z\mid u)\bigr)\Big],
  \label{eq:pop-kl-bound}
\end{equation}
which is the KL term used (up to the scalar weight $\beta$) in the training losses \eqref{eq:auxvae-loss} and \eqref{eq:dlcfm-loss}. In this sense, minimizing the training objective controls the mismatch between the conditional prior $p(z\mid u)$ and the encoder's aggregated posterior $q_\phi(z\mid u)$.

Second, the correlation-based penalties in~\eqref{eq:auxvae-loss} and \eqref{eq:dlcfm-loss} target the desired disentanglement properties at the level of encoder statistics (estimated via minibatches of the encoder mean $\mu_\phi(x)$). In the limit, as described in ~\cite{ganguli2025dlcfm,ganguli2025enhancing}, these penalties encourage:
\begin{subequations}
\begin{align}
  \text{(Explicitness)}\quad
    &\mathrm{Corr}(u_j,\mu_{\phi,\mathrm{aux},j}) \approx \pm 1,
    \qquad j=1,\dots,d,
    \\
  \text{(Intra-independence)}\quad
    &\mathrm{Corr}(u_j,\mu_{\phi,\mathrm{aux},j'}) \approx 0,
    \qquad j\neq j',
    \\
  \text{(Inter-independence)}\quad
    &\mathrm{Corr}(u,\mu_{\phi,\mathrm{rec}}) \approx 0.
\end{align}
\end{subequations}
These properties support an encoder-level interpretation: the first $d$ encoder coordinates can be treated as aligned with $u$ and the remaining coordinates as a residual block that is weakly dependent on $u$ under the aggregated posterior. This does not by itself imply a generator-side decoupling such as small tangent-space overlap $\rho_G$ in Definition~\ref{def:rhoG} (equivalently, the small-$\rho$ regime in Corollary~\ref{cor:block-diag}). In other words, Aux-VAE/DL-CFM regularizers act on statistics of $q_\phi(z\mid x)$ (e.g., encoder means and KL terms) and do not directly control Jacobian cross-sensitivities of $G_\theta$. Whenever we invoke $\rho_G\ll 1$ (or small observation-space cross-sensitivity $A_u^\top\Sigma_\varepsilon^{-1}A_{\mathrm{rec}}$ in Proposition~\ref{prop:identifiability}), we treat it as either (a) an additional empirical hypothesis to be validated by diagnostic measurements (e.g., principal angles between $\mathcal{T}_{\mathrm{aux}}$ and $\mathcal{T}_{\mathrm{rec}}$ on latent samples), or (b) a property to be explicitly enforced by augmenting training with a generator-side penalty that discourages cross-sensitivities, e.g.
\begin{equation}
  \mathcal{L}_{\mathrm{J\mbox{-}ortho}} \,\coloneqq\, \mathbb{E}_{z}\left[\left\|\left(\frac{\partial G_\theta}{\partial z_{\mathrm{aux}}}(z)\right)^\top\left(\frac{\partial G_\theta}{\partial z_{\mathrm{rec}}}(z)\right)\right\|_F^2\right].
  \label{eq:j-ortho}
\end{equation}

Corollary~\ref{cor:block-diag} and Proposition~\ref{prop:identifiability} suggest complementary, model-side diagnostics that connect training-time regularization to the behavior of the induced (linearized) prior/posterior: (i) estimate $\rho_G(z)=\|Q_{\mathrm{aux}}(z)^\top Q_{\mathrm{rec}}(z)\|_2$ on latent samples to quantify generator tangent overlap; (ii) estimate the observation cross-sensitivity $\|A_u^\top\Sigma_\varepsilon^{-1}A_{\mathrm{rec}}\|_2$ (or its normalized variant) at representative posterior points; and (iii) measure empirical posterior correlations between inferred $u$ and $z_{\mathrm{rec}}$ from latent-space sampling. These checks are inexpensive relative to retraining and help distinguish encoder-side disentanglement from generator-/posterior-level coupling.

\subsection{Latent-Space Inference}

Inference can be carried out in the low-dimensional latent space. Assuming a Gaussian noise model $\varepsilon \sim \mathcal{N}(0,\Sigma_\varepsilon)$ with known covariance, the log-posterior for the soft-constrained model~\eqref{eq:joint-posterior} is
\begin{align}
  \log p(u,z_{\mathrm{aux}},z_{\mathrm{rec}} \mid y)
  &=
  -\tfrac12\Bigl\|
    y - \mathcal{F}\bigl(G_\theta(z_{\mathrm{aux}},z_{\mathrm{rec}})\bigr)
  \Bigr\|_{\Sigma_\varepsilon^{-1}}^2
  -\tfrac1{2\tau^2}\|z_{\mathrm{aux}}-u\|_2^2
  + \log p(u)
  - \tfrac12\|z_{\mathrm{rec}}\|_2^2 + C,
  \label{eq:latent-logpost-full}
\end{align}
where $\|v\|_{A}^2 = v^\top A v$ and $C$ collects constants.

In the hard-constrained approximation $z_{\mathrm{aux}}=u$ (equivalently $\tau\to 0$), this reduces to inference in $(u,z_{\mathrm{rec}})$ with log-posterior
\begin{align}
  \log p(u,z_{\mathrm{rec}} \mid y)
  &=
  -\tfrac12\Bigl\|
    y - \mathcal{F}\bigl(G_\theta(u,z_{\mathrm{rec}})\bigr)
  \Bigr\|_{\Sigma_\varepsilon^{-1}}^2
  + \log p(u) - \tfrac12\|z_{\mathrm{rec}}\|_2^2 + C.
  \label{eq:latent-logpost}
\end{align}

Gradients with respect to $u$ and $z_{\mathrm{rec}}$ can be obtained via automatic differentiation through the generator $G_\theta$ and the forward model $\mathcal{F}$. This enables a range of inference algorithms. MAP estimation maximizes~\eqref{eq:latent-logpost} via gradient-based optimization (e.g., L-BFGS) to obtain a point estimate $(\hat u,\hat z_{\mathrm{rec}})$ and reconstruction $\hat x = G_\theta(\hat u,\hat z_{\mathrm{rec}})$. Variational inference chooses a variational family $q_\psi(u,z_{\mathrm{rec}})$, such as a mean-field Gaussian or a low-rank Gaussian, and optimizes an evidence lower bound. MCMC or related algorithms in latent space generate approximate posterior samples of $(u,z_{\mathrm{rec}})$, and push them through $G_\theta$ to obtain samples in the $x$-space. Because the latent dimension $d_Z$ is typically much smaller than the dimension of $x$, these methods avoid sampling directly in the full state space. Whether this yields a computational gain still depends on the cost of evaluating and differentiating $\mathcal{F}\circ G_\theta$.

\subsection{Computational Workflow}

We summarize the computational workflow for constructing and using a disentangled deep prior.

\paragraph{Offline training stage.} Suppose we have access to a dataset of fields $\{x_i\}_{i=1}^N$, either from simulations or from high-quality measurements, along with associated auxiliary variables $\{u_i\}_{i=1}^N$ that encode known generative factors. The offline stage consists of: (1) selecting either an Aux-VAE or a DL-CFM model with a latent partition $z = (z_{\mathrm{aux}},z_{\mathrm{rec}})$ and conditional prior $p(z\mid u)$ as in~\eqref{eq:aux-prior}; (2) training the model by minimizing either $\mathcal{L}_{\mathrm{AuxVAE}}$~\eqref{eq:auxvae-loss} or $\mathcal{L}_{\mathrm{DL\mbox{-}CFM}}$~\eqref{eq:dlcfm-loss}, using minibatch estimates of the correlation-based regularizers; and (3) assessing whether $z_{\mathrm{aux}}$ has the desired alignment with $u$ and whether $z_{\mathrm{rec}}$ behaves as a residual block. Relevant diagnostics include scatter plots of $u_j$ vs.\ $z_{\mathrm{aux},j}$, dependence metrics such as the linear disentanglement score, latent traversals where $z_{\mathrm{aux}}$ or $z_{\mathrm{rec}}$ is perturbed while the other block is fixed, and generator-side checks such as estimating $\rho_G(z)$ (Definition~\ref{def:rhoG}) or $\|J_{\mathrm{aux}}(z)^\top J_{\mathrm{rec}}(z)\|$ on latent samples.

\paragraph{Prior definition.} After training, we freeze the generator $G_\theta$ and the latent prior structure, and use~\eqref{eq:prior-u}--\eqref{eq:prior-x} as a prior in the inverse problem. The prior over $u$ may be chosen according to physics-based considerations: a Gaussian prior on coarse summary statistics, a uniform prior on physically admissible ranges, or a mixture prior representing different regimes (e.g., multiple operational modes of a device).

\paragraph{Online inversion.} Given observed data $y$, we can infer $p(u,z_{\mathrm{rec}}\mid y)$ using MAP, variational methods, or latent-space MCMC. Summaries for $u$ (interpretable) and for residual variability induced by $z_{\mathrm{rec}}$ can be used post-hoc.

\section{Numerical Experiments: Elliptic PDE Inverse Problems}
\label{sec:experiments_2d}

We evaluate the proposed approach on two 2D steady-state heat problems on the
unit square $\Omega=[0,1]^2$, discretized on a uniform $28\times 28$ grid with
$676$ interior degrees of freedom. The unknown is a positive field represented
in log-coordinates, $x=\log\kappa$ (conductivity) or $x=\log b$ (source), so
that $\kappa=\exp(x)$ or $b=\exp(x)$ elementwise. Observations are noisy
temperature values on a subset of grid points.

For each problem we compare four inference strategies:
\begin{enumerate}
  \item \textbf{AuxVAE} (proposed): Structured VAE prior with disentangled
        auxiliary latent variables and soft-constrained posterior.
  \item \textbf{Plain VAE}: Standard VAE prior without auxiliary structure
        (ELBO training only).
  \item \textbf{GP-Fixed}: Gaussian process prior with oracle (true)
        hyperparameters --- a best-case baseline.
  \item \textbf{GP-Hier}: Hierarchical GP with free hyperparameters sampled
        jointly with the whitened field.
\end{enumerate}

\subsection{Forward model and observation model}
\label{sec:forward_model}

For the conductivity problem, given $\kappa(s)>0$ we solve
\begin{equation}
  -\nabla\cdot\bigl(\kappa(s)\,\nabla T(s)\bigr) = 50
  \quad \text{in }\Omega,
  \qquad T\big|_{\partial\Omega} = 0,
  \label{eq:pde_cond}
\end{equation}
with homogeneous Dirichlet boundary conditions and a constant volumetric
source. Let $T(\kappa)\in\mathbb{R}^{p}$ denote the discrete temperature field.

For the source identification problem, given $b(s)>0$ we solve
\begin{equation}
  -\Delta T(s) = b(s)
  \quad \text{in }\Omega,
  \qquad T\big|_{\partial\Omega} = 0,
  \label{eq:pde_source}
\end{equation}
with known unit conductivity ($\kappa\equiv 1$) and unknown source field,
and denote the discrete solution by $T(b)\in\mathbb{R}^{p}$.

Let $M\in\{0,1\}^{m\times p}$ be a masking matrix that selects observed grid
points. We observe
\begin{equation}
  y = M\,T(x) + \varepsilon,
  \qquad \varepsilon \sim \mathcal{N}(0,\,\sigma^2 I_m),
  \label{eq:obs_model}
\end{equation}
where $T(x)$ denotes the discrete PDE solution after mapping from $x$ to the
physical coefficient. The negative log-likelihood is proportional to
\begin{equation}
  \Phi(x;\,y) \;:=\; \frac{1}{2\sigma^2}\,\|y - M\,T(x)\|_2^2.
  \label{eq:nll}
\end{equation}

\subsection{Disentangled deep prior and latent-space inference}
\label{sec:prior_inference}

We train an Aux-VAE offline on pairs $\{(x_i,u_i)\}$, where
$u\in\mathbb{R}^d$ are interpretable dataset-specific factors. The latent is
partitioned as $z=(z_{\mathrm{aux}},z_{\mathrm{rec}})$ with
$\dim(z_{\mathrm{aux}})=d$ and total dimension $d_Z$. During training we use a
soft-tether conditional prior
\begin{equation}
  p(z \mid u_{\mathrm{norm}}) =
  \mathcal{N}\!\left(
    \begin{bmatrix} u_{\mathrm{norm}} \\ 0 \end{bmatrix},\;
    \operatorname{diag}(\tau^2 I_d,\, I_{d_Z-d})
  \right),
  \qquad \tau_{\mathrm{train}}=0.1,
  \label{eq:soft_tether}
\end{equation}
where $u_{\mathrm{norm}}$ denotes the normalized auxiliary variables used in
training. We place a standard normal prior on the normalized auxiliary
variables, $p(u_{\mathrm{norm}})=\mathcal{N}(0,I_d)$. The posterior density
over $(u_{\mathrm{norm}},z_{\mathrm{rec}})$ is
\begin{equation}
  p(u_{\mathrm{norm}},z_{\mathrm{rec}}\mid y)
  \;\propto\;
  \exp\!\bigl(-\Phi(G_\theta(u_{\mathrm{norm}},z_{\mathrm{rec}});\,y)\bigr)\,
  \exp\!\bigl(-\tfrac{1}{2}\|u_{\mathrm{norm}}\|_2^2\bigr)\,
  \exp\!\bigl(-\tfrac{1}{2}\|z_{\mathrm{rec}}\|_2^2\bigr).
  \label{eq:posterior}
\end{equation}

We compute a MAP estimate by gradient-based optimization of the negative
log-posterior, and draw posterior samples using Hamiltonian Monte Carlo (HMC)
in the $d_Z$-dimensional latent space. Because $d_Z \ll p$, HMC is
computationally feasible even though the unknown field is high-dimensional.

\subsection{Problem-specific setup}
\label{sec:setup_2d}

\paragraph{Conductivity problem.}
The forward model is the steady-state diffusion equation~\eqref{eq:pde_cond}
with unknown log-conductivity $\log\kappa$.
Training fields are sampled from an anisotropic squared-exponential GP
parameterised by $\mathbf{u}=(\mu,\sigma,\ell_x,\ell_y)$, with
$\mu\!\in\![-0.5,0.5]$, $\sigma\!\in\![0.2,0.7]$,
$\ell_x,\ell_y\!\in\![0.05,0.35]$.
We generate $N_{\mathrm{train}}=12{,}000$ samples
(Figure~\ref{fig:training_history}a shows the training loss curves).
The test case uses $\mu^*\!=\!0.30$, $\sigma^*\!=\!0.40$,
$\ell_x^*\!=\!0.20$, $\ell_y^*\!=\!0.15$.
Observations are collected at $34$ randomly placed interior sensors ($5\%$
coverage) with additive Gaussian noise $\sigma_{\mathrm{obs}}=0.50$.
This is a \emph{well-specified} setting for the GP baselines, since the
true field is indeed a GP draw.

\paragraph{Source identification problem.}
The forward model is the Poisson equation~\eqref{eq:pde_source}
with known $\kappa=1$ and unknown log-source field $\log b$.
Each training field is constructed as a dominant Gaussian bump
$b(\mathbf{s}) = A\exp\bigl(-\|\mathbf{s}-\mathbf{c}\|^2/(2w^2)\bigr)$
plus small random residual bumps, with auxiliary parameters
$\mathbf{u}=(A,c_x,c_y,w)$.
Ranges: $A\!\in\![20,80]$, $c_x,c_y\!\in\![0.2,0.8]$, $w\!\in\![0.04,0.12]$.
We use $N_{\mathrm{train}}=5{,}000$ samples following the DeepGenPrior setup
(Figure~\ref{fig:training_history}b).
The test case uses $A^*\!=\!50$, $c_x^*\!=\!0.45$, $c_y^*\!=\!0.55$,
$w^*\!=\!0.08$.
Observations: $202$ sensors ($30\%$ interior coverage),
$\sigma_{\mathrm{obs}}=0.01$.
This is a \emph{misspecified} setting for the GP baselines: the true source
is a localised Gaussian bump, not a GP draw.

\paragraph{AuxVAE architecture.}
The encoder uses a Conv1D backbone (base channels $24$ for conductivity,
$16$ for source, depth~$3$) with adaptive max-pooling.
The decoder is an MLP.
Latent dimensions: $d_{\mathrm{aux}}=4$ (matching the physical parameters),
$d_{\mathrm{rec}}=48$ (conductivity) or $16$ (source).
Training uses the disentangled ELBO with polynomial-lift penalties
($\lambda_1,\lambda_2,\lambda_3$) and a $200$-epoch reconstruction
warmup followed by end-to-end optimisation.
Figure~\ref{fig:disentanglement} shows the learned auxiliary--parameter
correspondence for both problems.

\begin{figure}[htbp]
\centering
\begin{subfigure}[t]{0.48\textwidth}
  \centering
  \includegraphics[width=\textwidth]{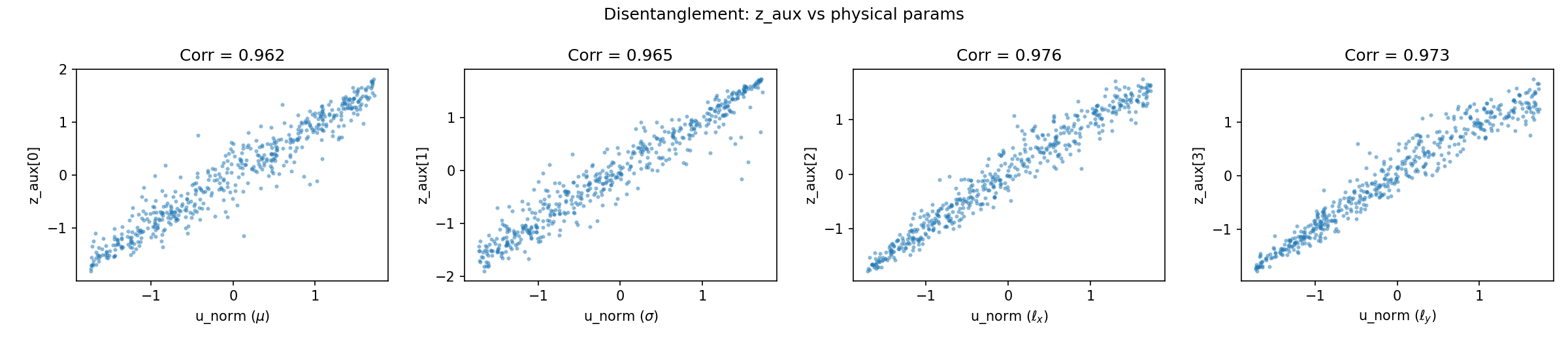}
  \caption{Conductivity}
\end{subfigure}\hfill
\begin{subfigure}[t]{0.48\textwidth}
  \centering
  \includegraphics[width=\textwidth]{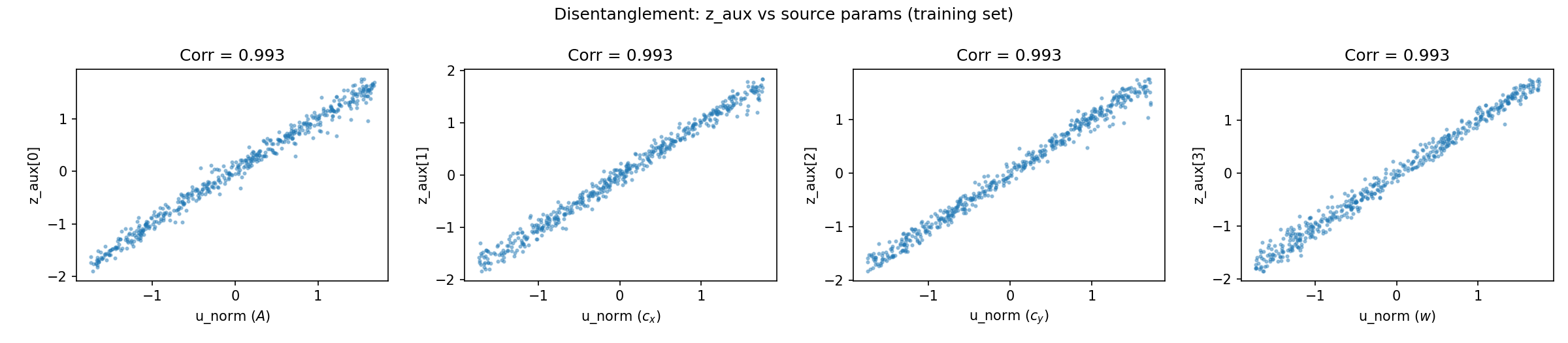}
  \caption{Source}
\end{subfigure}
\caption{Disentanglement scatter plots: learned auxiliary latent variables
  vs.\ true physical parameters for (a)~conductivity and (b)~source problems.
  Generator-level cross-sensitivity is moderate in both cases
  ($J_{\mathrm{ortho,norm}} = 0.139$ for conductivity, $0.086$ for source).
  Posterior-level decoupling, measured as mean $|\mathrm{Corr}(u_i, z_{\mathrm{rec},j})|$
  from HMC samples, is $0.030$ (max $0.57$) for conductivity and
  $0.147$ (max $0.47$) for source, indicating moderate to weak residual coupling
  between the interpretable and residual latent blocks.}

\label{fig:disentanglement}
\end{figure}

\begin{figure}[htbp]
\centering
\begin{subfigure}[t]{0.48\textwidth}
  \centering
  \includegraphics[width=\textwidth]{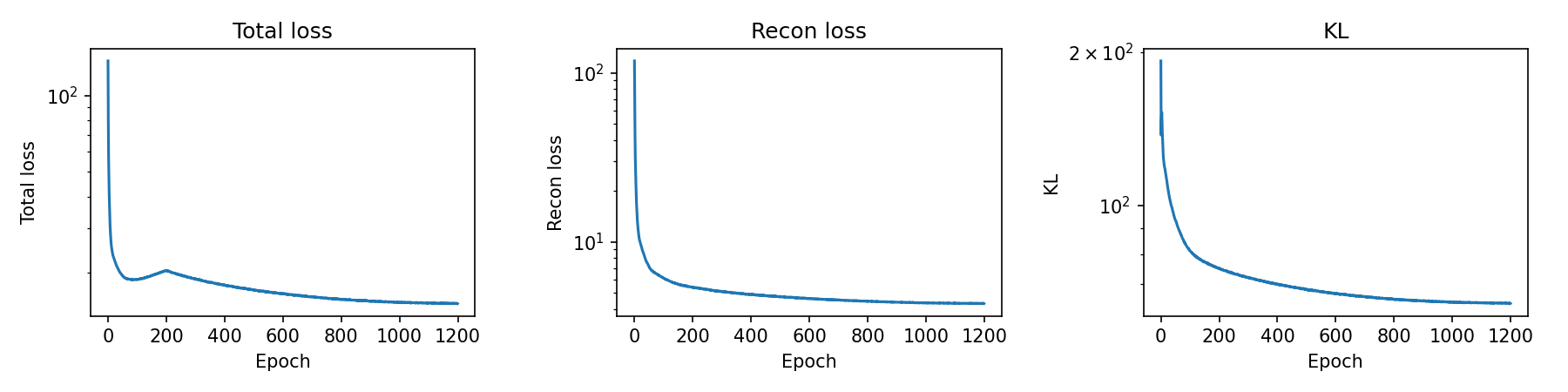}
  \caption{Conductivity}
\end{subfigure}\hfill
\begin{subfigure}[t]{0.48\textwidth}
  \centering
  \includegraphics[width=\textwidth]{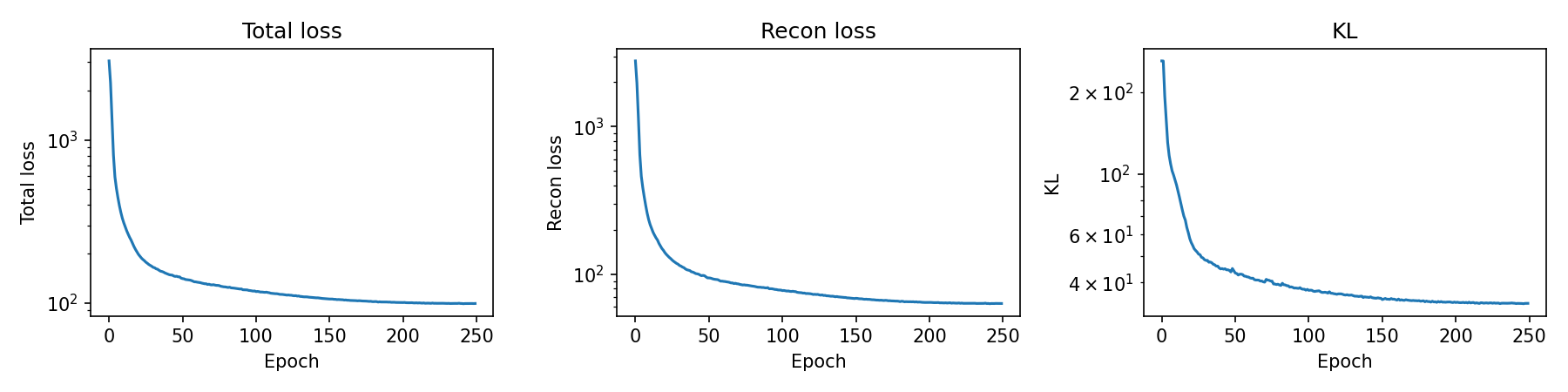}
  \caption{Source}
\end{subfigure}
\caption{AuxVAE training loss curves for both problems.}
\label{fig:training_history}
\end{figure}

\paragraph{HMC configuration.}
All methods use HMC with dual-averaging step-size adaptation during a
$2{,}000$-step warmup, followed by $3{,}000$ post-warmup samples.
Two-phase diagonal mass-matrix adaptation is used: during the first half
of warmup, samples are collected under $M=I$ to estimate per-coordinate
variances, then $M$ is set to a diagonal matrix matching these variances
and the step-size adaptation is reset.
Trajectory lengths are tuned per problem:
$L=15$ for conductivity ($\varepsilon_0=0.15$) and
$L=150$ for source ($\varepsilon_0=0.01$), reflecting the sharper posterior
geometry induced by the low observation noise in the source problem.
GP baselines use smaller step sizes
($\varepsilon_0=0.03$ for GP-Fixed, $\varepsilon_0=0.01$ for GP-Hier
in conductivity; $\varepsilon_0=0.01$ and $0.005$ in source)
due to the higher-dimensional state space ($784$ field DOFs).

\paragraph{Effective sample size.}
We report the mean effective sample size (ESS) computed via FFT-based
autocorrelation with truncation at the first non-positive autocorrelation
lag~\citep{geyer1992practical}.
ESS is computed over the \emph{full sampling space} for each method:
the complete latent chain $(\mathbf{z}_{\mathrm{aux}}, \mathbf{z}_{\mathrm{rec}})$
for the AuxVAE and Plain VAE ($d=52$ for conductivity, $d=20$ for source),
the whitened field $\boldsymbol{\xi}$ for GP-Fixed ($d=784$), and
$(\boldsymbol{\xi}, \boldsymbol{\theta})$ for GP-Hier ($d=788$).

\subsection{Results}
\label{sec:results_2d}

\subsubsection{Conductivity Identification (Well-Specified Prior)}

Table~\ref{tab:cond_results} summarises the HMC results.
All four methods achieve low field RMSE ($0.34$--$0.36$) and
well-calibrated $95\%$ credible intervals (coverage $0.91$--$0.98$).
GP-Fixed performs best in field accuracy (RMSE~$=0.340$, coverage~$=0.978$),
as expected for an oracle baseline with known hyperparameters.
AuxVAE achieves comparable accuracy (RMSE~$=0.358$) with excellent
calibration (coverage~$=0.969$) and recovers all four GP hyperparameters
within their $95\%$ credible intervals.
GP-Hier also recovers all parameters but suffers from very low acceptance
($0.037$) and low ESS~($70$), reflecting the difficulty of jointly sampling
hyperparameters and the $784$-dimensional whitened field with HMC.
The AuxVAE achieves the highest mean ESS ($1{,}046$) among the methods
with learnable hyperparameters, with a $7\times$ improvement over Plain VAE
($150$) and a $15\times$ improvement over GP-Hier ($70$).

\begin{table}[htbp]
\centering
\caption{Conductivity identification results (HMC, $L=20$, $n=3{,}000$ samples, warmup $=500$).}
\label{tab:cond_results}
\begin{tabular}{lccccc}
\toprule
Method & Field RMSE & 95\% Cov. & Accept & ESS (mean) & Params in CI \\
\midrule
AuxVAE (ours)   & 0.358 & 0.969 & \textbf{0.869} & \textbf{1046} & 4/4 \\
Plain VAE       & 0.350 & 0.955 & 0.114          & 150           & --- \\
GP-Fixed        & \textbf{0.340} & \textbf{0.978} & 0.581 & 3000$^\dagger$ & --- (oracle) \\
GP-Hier         & 0.358 & 0.911 & 0.037          & 70            & 4/4 \\
\bottomrule
\multicolumn{6}{l}{\footnotesize $^\dagger$\,GP-Fixed operates in a pre-whitened space with oracle hyperparameters,}\\
\multicolumn{6}{l}{\footnotesize yielding a near-Gaussian posterior and essentially uncorrelated samples.}
\end{tabular}
\end{table}

\paragraph{Posterior parameter recovery.}
The AuxVAE posterior means and standard deviations for the GP
hyperparameters are:
$\mu = 0.27 \pm 0.22$ (true $0.30$),
$\sigma = 0.45 \pm 0.14$ (true $0.40$),
$\ell_x = 0.20 \pm 0.09$ (true $0.20$),
$\ell_y = 0.20 \pm 0.09$ (true $0.15$).
All true values fall well within the posterior credible intervals
(Figure~\ref{fig:cond_marginals}).

\begin{figure}[htbp]
\centering
\begin{subfigure}[t]{0.48\textwidth}
  \centering
  \includegraphics[width=\textwidth]{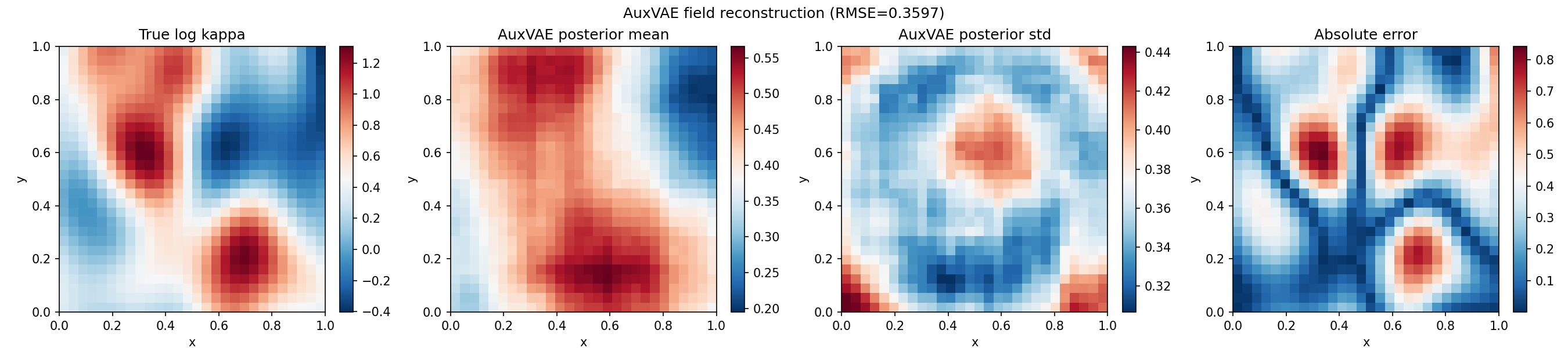}
  \caption{AuxVAE: temperature field}
\end{subfigure}\hfill
\begin{subfigure}[t]{0.48\textwidth}
  \centering
  \includegraphics[width=\textwidth]{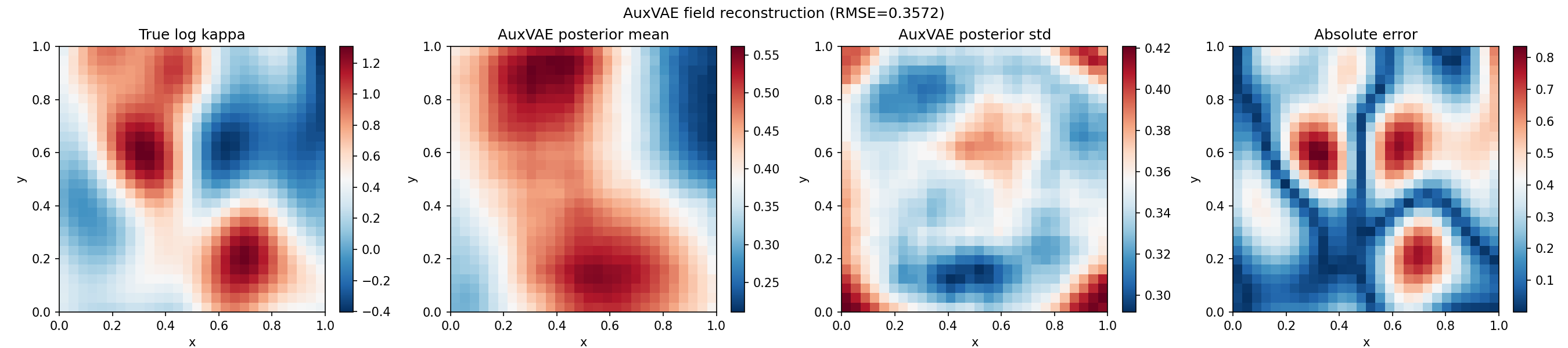}
  \caption{AuxVAE: conductivity field}
\end{subfigure}\\[6pt]
\begin{subfigure}[t]{0.48\textwidth}
  \centering
  \includegraphics[width=\textwidth]{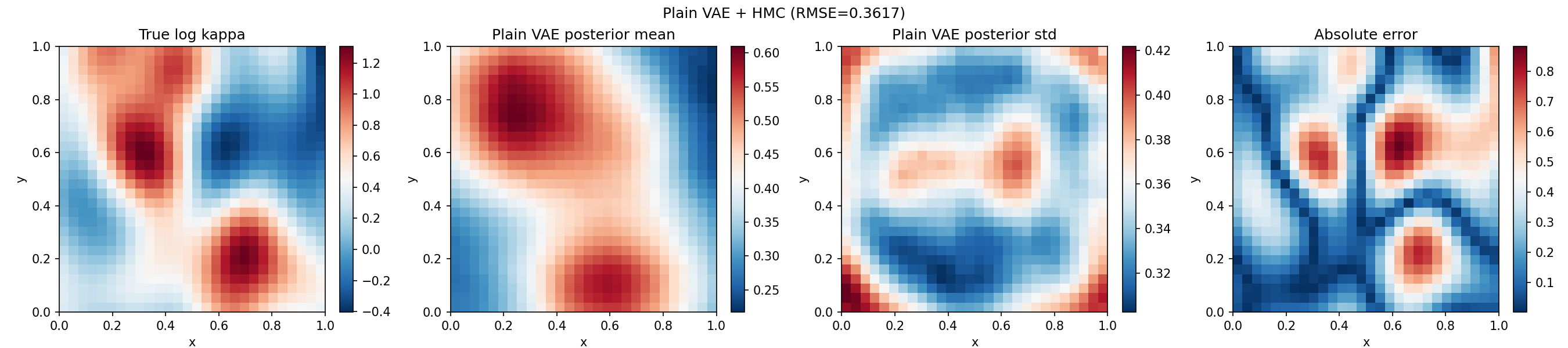}
  \caption{Plain VAE: temperature field}
\end{subfigure}\hfill
\begin{subfigure}[t]{0.48\textwidth}
  \centering
  \includegraphics[width=\textwidth]{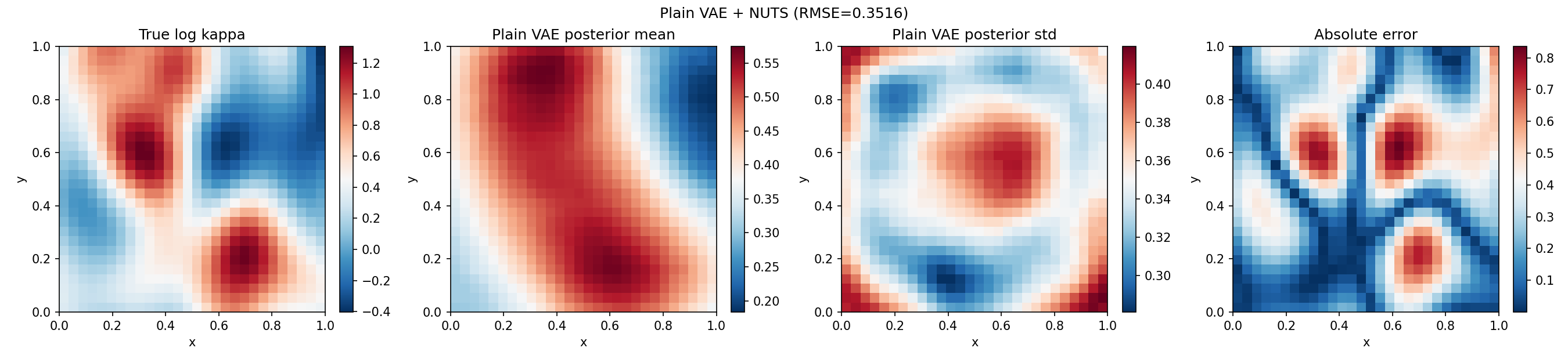}
  \caption{Plain VAE: conductivity field}
\end{subfigure}\\[6pt]
\begin{subfigure}[t]{0.48\textwidth}
  \centering
  \includegraphics[width=\textwidth]{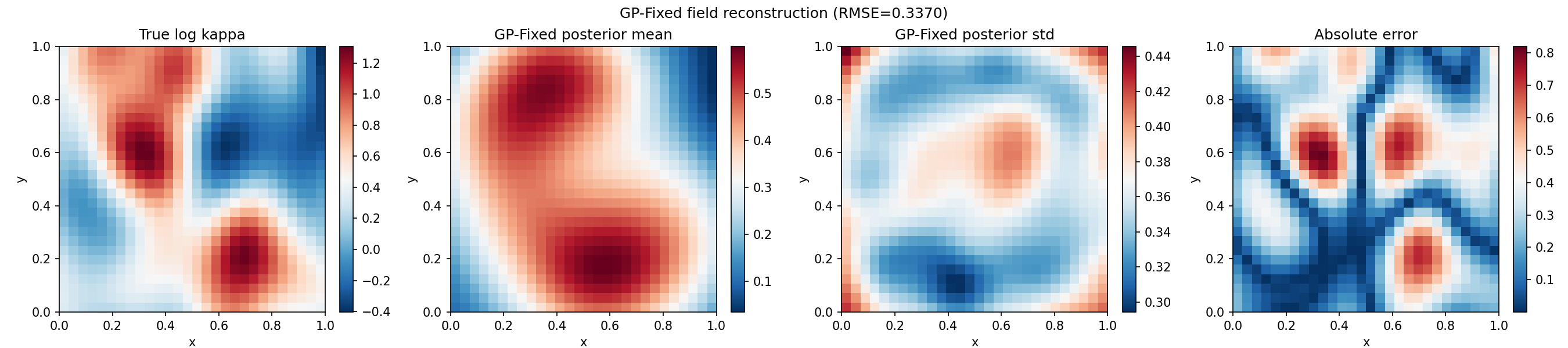}
  \caption{GP-Fixed: temperature field}
\end{subfigure}\hfill
\begin{subfigure}[t]{0.48\textwidth}
  \centering
  \includegraphics[width=\textwidth]{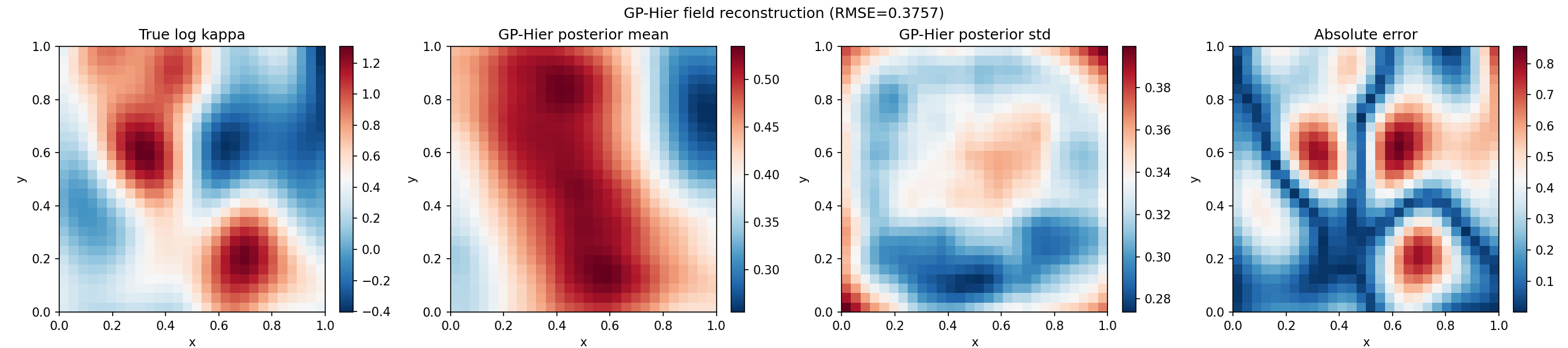}
  \caption{GP-Hier: temperature field}
\end{subfigure}
\caption{Conductivity identification: posterior mean field reconstructions
  for all four methods. True field, posterior mean, and pointwise standard
  deviation are shown.}
\label{fig:cond_fields}
\end{figure}

\begin{figure}[htbp]
\centering
\includegraphics[width=0.7\textwidth]{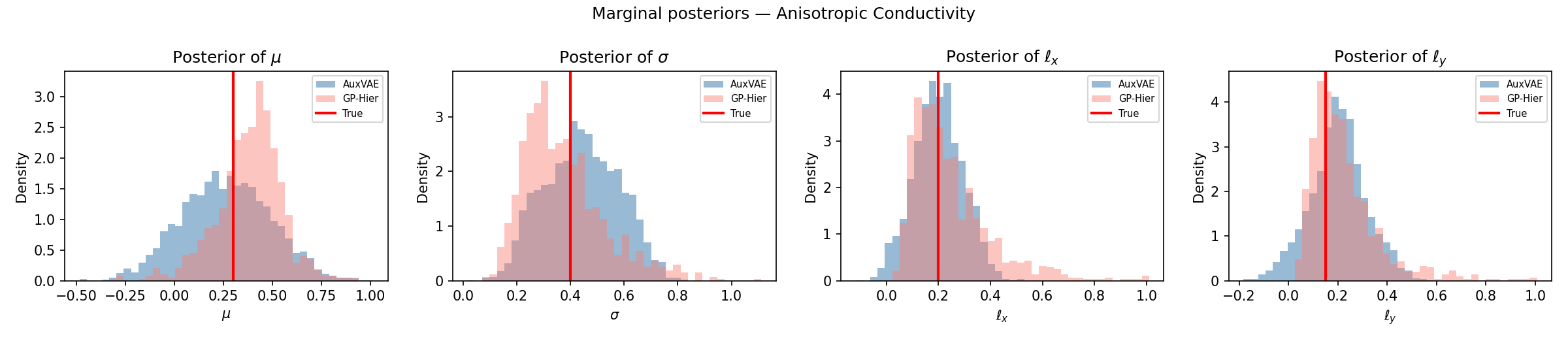}
\caption{Conductivity identification: marginal posterior distributions of the
  GP hyperparameters $(\mu,\sigma,\ell_x,\ell_y)$ from the AuxVAE.
  Vertical dashed lines indicate true values.}
\label{fig:cond_marginals}
\end{figure}

\begin{figure}[htbp]
\centering
\includegraphics[width=0.7\textwidth]{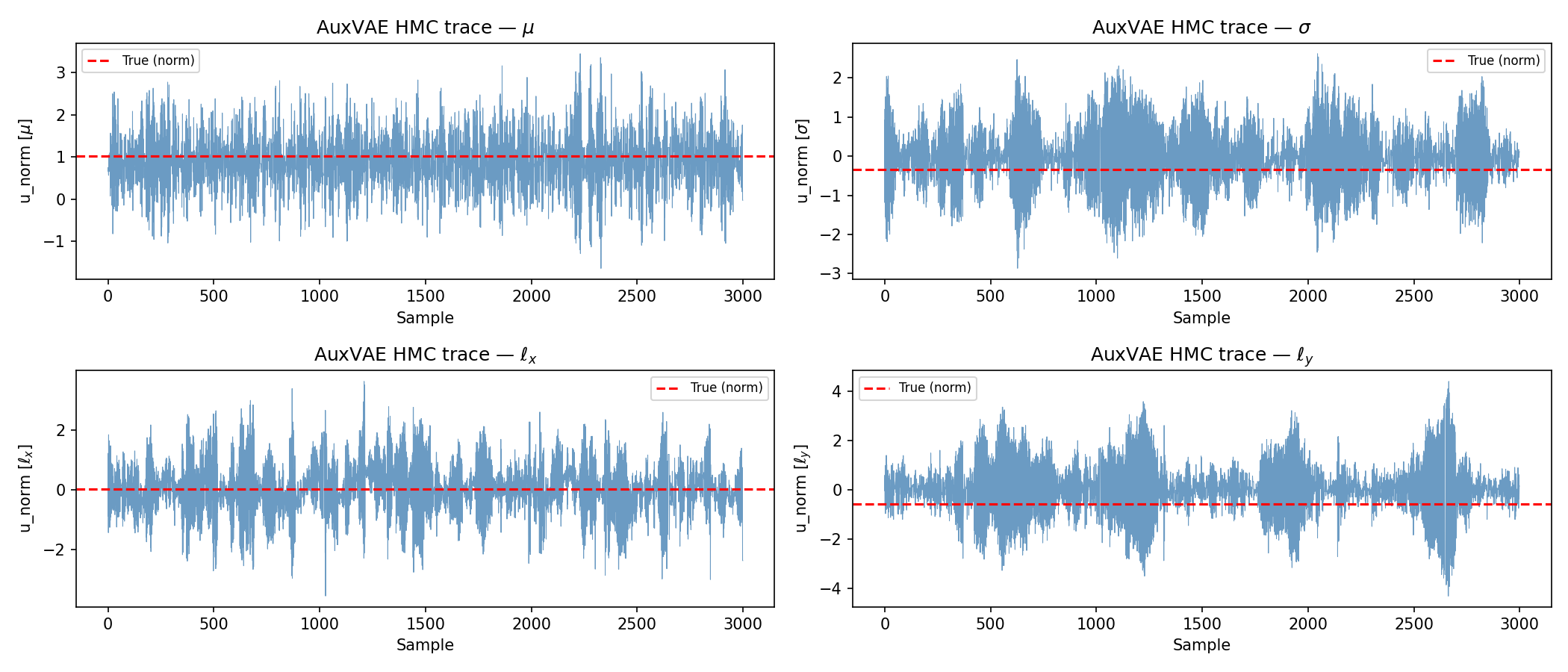}
\caption{Conductivity identification: HMC trace plots for the AuxVAE
  auxiliary latent dimensions.}
\label{fig:cond_traces}
\end{figure}

\subsubsection{Temperature Field Reconstruction}

Figure~\ref{fig:y_true_vs_y_pred_T} compares the true and predicted
temperature fields for both problems.
In the conductivity problem (sparse observations, $5\%$ coverage,
$\sigma_{\mathrm{obs}}=0.50$), the AuxVAE achieves $R^2 = 0.977$ and
$\mathrm{Corr} = 0.989$, with the scatter around the diagonal reflecting
the moderate observation noise and limited sensor coverage.
In the source problem (dense observations, $30\%$ coverage,
$\sigma_{\mathrm{obs}}=0.01$), the reconstruction is essentially
perfect ($R^2 = 0.9999$, $\mathrm{Corr} = 1.000$).
These results confirm that the AuxVAE prior enables accurate recovery of
the observable temperature field in both regimes.

\begin{figure}[htbp]
\centering
\includegraphics[width=\textwidth]{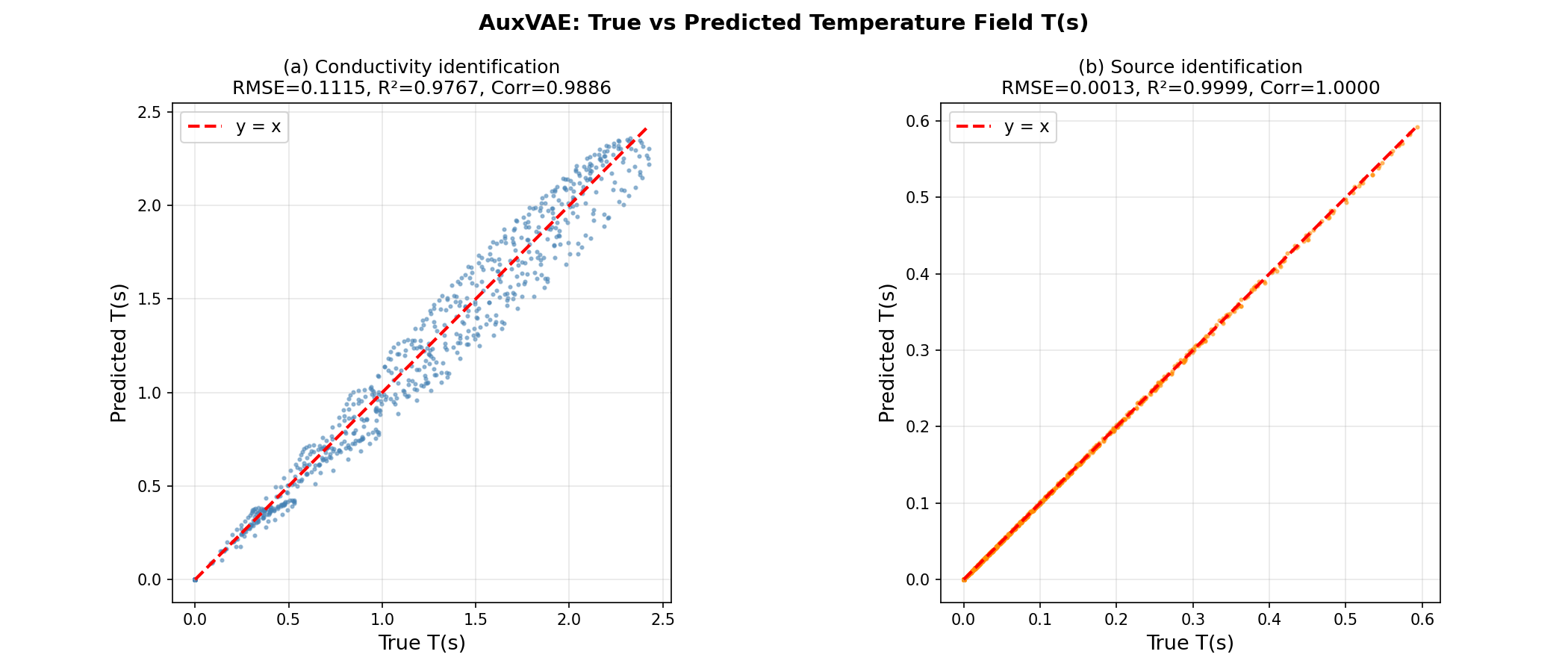}
\caption{True vs.\ predicted temperature field $T(\mathbf{s})$ from the
  AuxVAE posterior mean for (a)~conductivity identification and
  (b)~source identification. The dashed red line indicates perfect
  reconstruction ($y = x$).}
\label{fig:y_true_vs_y_pred_T}
\end{figure}

\subsubsection{Source Identification (Misspecified Prior)}

Table~\ref{tab:source_results} presents the source identification results.
Here the GP baselines face a fundamental model mismatch: the true source
field is a localised Gaussian bump, not a realisation from a GP.

The AuxVAE achieves the best RMSE ($1.455$) and near-nominal coverage
($0.944$), recovering all four source parameters within their credible
intervals.
The Plain VAE obtains higher RMSE ($1.729$) and substantially lower coverage ($0.677$),
highlighting the value of the structured auxiliary latent space for
calibration.
GP-Fixed, despite oracle hyperparameters, yields RMSE~$=4.484$ and
coverage of only $0.241$ --- the GP kernel cannot represent the localised
bump structure.
GP-Hier effectively fails, with coverage~$\approx 21\%$, ESS~$\approx 6$, and
$0/4$ parameters recovered.

\begin{table}[htbp]
\centering
\caption{Source identification results (HMC, $L=100$, $n=3{,}000$ samples, warmup $=500$).}
\label{tab:source_results}
\begin{tabular}{lccccc}
\toprule
Method & Field RMSE & 95\% Cov. & Accept & ESS (mean) & Params in CI \\
\midrule
AuxVAE (ours)   & \textbf{1.455} & \textbf{0.944} & \textbf{0.813} & 66   & 4/4 \\
Plain VAE       & 1.729          & 0.677          & 0.316          & 104  & --- \\
GP-Fixed        & 4.484          & 0.241          & 0.927          & 111  & --- (oracle) \\
GP-Hier         & 3.405          & 0.213          & 0.430          & 6    & 0/4 \\
\bottomrule
\end{tabular}
\end{table}

\paragraph{Posterior parameter recovery.}
The AuxVAE recovers the source parameters
(Figure~\ref{fig:source_marginals}):
$A = 61.5 \pm 17.7$ (true $50.0$),
$c_x = 0.438 \pm 0.069$ (true $0.45$),
$c_y = 0.598 \pm 0.088$ (true $0.55$),
$w = 0.082 \pm 0.008$ (true $0.080$).
The width $w$ and centres $(c_x, c_y)$ are tightly constrained, while the
amplitude $A$ has wider uncertainty, consistent with the PDE's smoothing
effect on the temperature field.

\begin{figure}[htbp]
\centering
\begin{subfigure}[t]{0.48\textwidth}
  \centering
  \includegraphics[width=\textwidth]{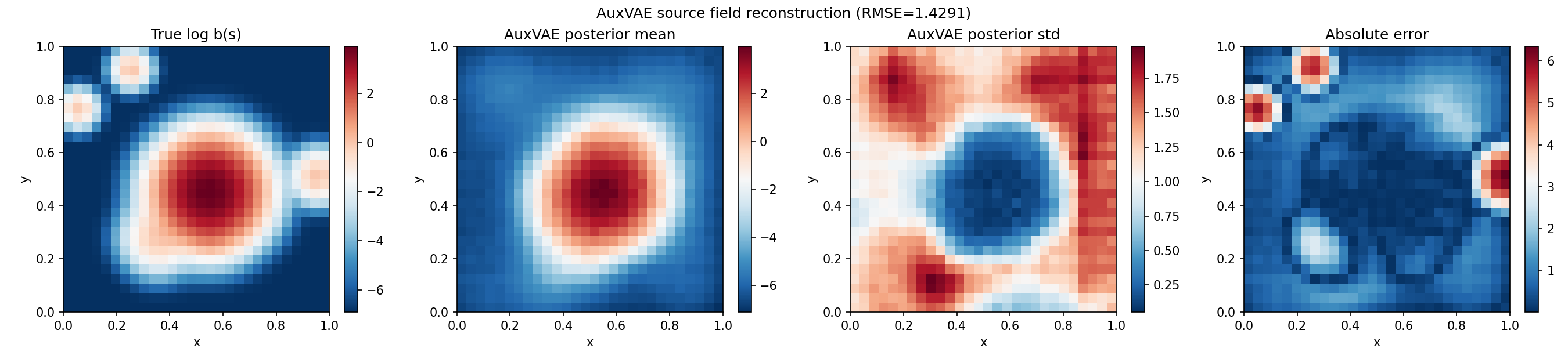}
  \caption{AuxVAE}
\end{subfigure}\hfill
\begin{subfigure}[t]{0.48\textwidth}
  \centering
  \includegraphics[width=\textwidth]{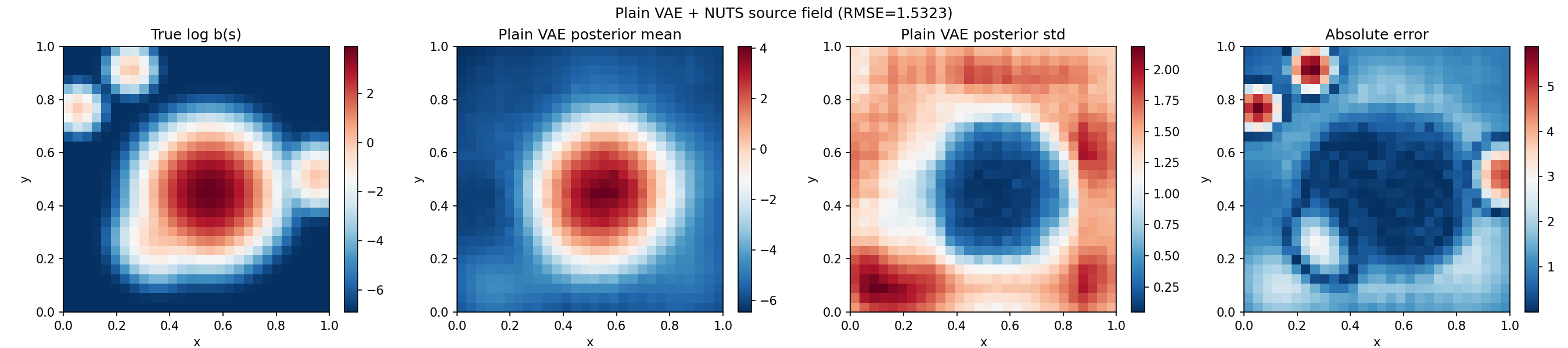}
  \caption{Plain VAE}
\end{subfigure}\\[6pt]
\begin{subfigure}[t]{0.48\textwidth}
  \centering
  \includegraphics[width=\textwidth]{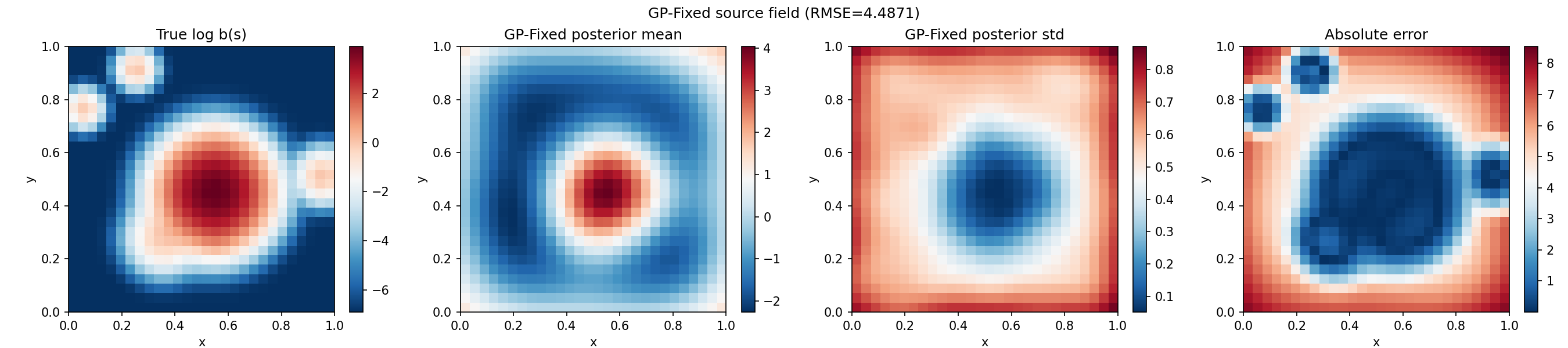}
  \caption{GP-Fixed}
\end{subfigure}\hfill
\begin{subfigure}[t]{0.48\textwidth}
  \centering
  \includegraphics[width=\textwidth]{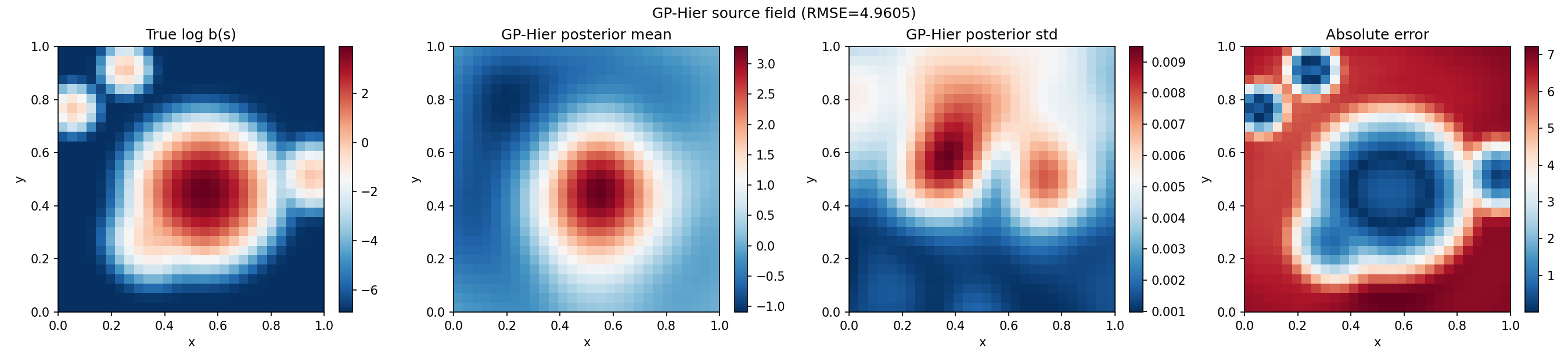}
  \caption{GP-Hier}
\end{subfigure}
\caption{Source identification: posterior mean field reconstructions for all
  four methods. The localised Gaussian bump structure is well captured by the
  AuxVAE but poorly represented by the GP baselines.}
\label{fig:source_fields}
\end{figure}

\paragraph{Spatially aware uncertainty.}
While none of the four methods recover the small residual bumps in the
posterior mean (Figure~\ref{fig:source_fields}), the AuxVAE is the only
method whose posterior standard deviation is elevated precisely at the
locations of these residual features.
This indicates that the AuxVAE prior has learned enough about the field
distribution to recognise that localised deviations from the dominant
Gaussian structure are plausible, and it expresses this as spatially
concentrated uncertainty rather than a uniform error.
By contrast, the GP baselines show no spatially structured uncertainty
correlated with the residual bumps --- their posterior variance is either
globally diffuse (GP-Fixed) or collapsed (GP-Hier).
This spatially aware uncertainty is a hallmark of well-calibrated
Bayesian inference and is reflected in the AuxVAE's near-nominal
$95\%$ coverage ($0.944$), which requires the credible intervals to be
wider where the true field deviates from the posterior mean and tighter
where it does not.

\begin{figure}[htbp]
\centering
\includegraphics[width=0.7\textwidth]{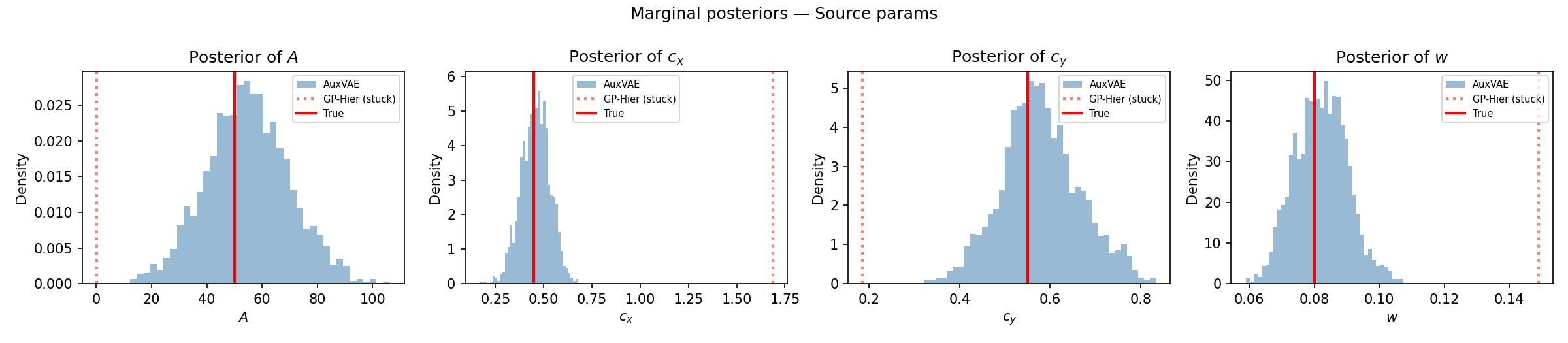}
\caption{Source identification: marginal posterior distributions of the source
  parameters $(A, c_x, c_y, w)$ from the AuxVAE. GP-Hier is shown as a
  vertical line (stuck chain, ESS$\,\approx 6$). Vertical dashed lines indicate
  true values.}
\label{fig:source_marginals}
\end{figure}

\begin{figure}[htbp]
\centering
\includegraphics[width=0.7\textwidth]{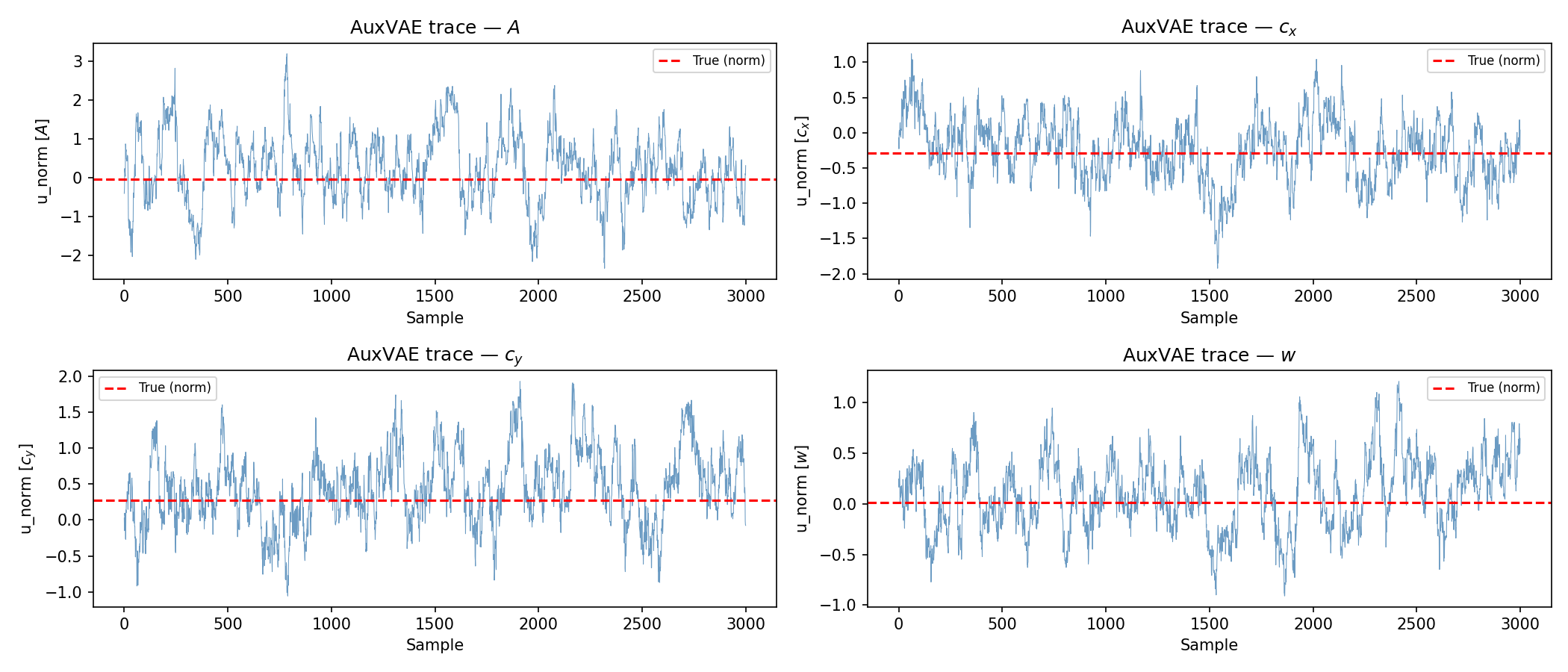}
\caption{Source identification: HMC trace plots for the AuxVAE auxiliary
  latent dimensions.}
\label{fig:source_traces}
\end{figure}

\subsubsection{Why Do the Two Problems Yield Such Different Results?}

Comparing Tables~\ref{tab:cond_results} and~\ref{tab:source_results} reveals
striking differences: in conductivity all methods achieve RMSE~$<0.36$ and
coverage~$>0.91$, while in source identification the GP baselines collapse to
RMSE~$>3.4$ and coverage~$<25\%$.
Three factors explain this contrast.

\paragraph{Prior specification.}
The conductivity problem is \emph{well-specified} for the GP baselines:
the true field is literally a draw from the same anisotropic SE kernel
used by GP-Fixed and GP-Hier.
In this setting, the GP prior provides the correct inductive bias, and even
the oracle GP-Fixed baseline has access to the exact generative
hyperparameters.
By contrast, the source problem is \emph{misspecified}: the true field is a
localised Gaussian bump with compact support, a structure that no
stationary GP kernel can efficiently represent.
A GP prior spreads probability mass over globally smooth realisations and
cannot concentrate on the peaked, localised fields characteristic of the
source problem.
This mismatch is irreparable regardless of hyperparameter tuning ---
GP-Fixed with oracle hyperparameters still achieves only $24\%$ coverage.

\paragraph{Role of the learned prior.}
The AuxVAE sidesteps the specification issue entirely by learning the
field distribution from training data.
In conductivity, the AuxVAE implicitly learns the GP structure and matches
the oracle baseline.
In source identification, the AuxVAE learns the Gaussian bump manifold
parameterised by $(A, c_x, c_y, w)$, enabling it to represent the
localised structure that the GP cannot.
This is the central advantage of a data-driven prior: it adapts to the
problem structure without requiring the practitioner to specify the correct
parametric family.

\paragraph{Observation regime.}
The two problems also differ in their observation configurations.
The conductivity problem uses sparse observations ($5\%$ coverage) with
moderate noise ($\sigma_{\mathrm{obs}}=0.50$), creating a broad posterior
that is relatively easy to sample.
The source problem uses denser observations ($30\%$ coverage) with low noise
($\sigma_{\mathrm{obs}}=0.01$), creating a sharp, tightly concentrated
posterior.
This sharp posterior forces the HMC step size to small values
($\varepsilon \approx 0.005$--$0.01$) and requires longer trajectories
($L=100$ vs $L=20$ for conductivity) to achieve adequate mixing, resulting
in lower ESS across all methods.

\paragraph{Dimensionality and posterior geometry.}
The GP baselines sample in the full $784$-dimensional field space (plus
$4$ hyperparameters for GP-Hier), while the AuxVAE operates in a
$20$--$52$-dimensional latent space.
This dimensionality reduction, combined with the disentangled structure of
the auxiliary variables, yields a better-conditioned posterior that HMC can
traverse efficiently.
GP-Hier faces the additional challenge of a \emph{funnel geometry} where
the hyperparameters and the field are strongly coupled: small changes in
length scales $(\ell_x, \ell_y)$ reshape the entire kernel matrix, creating
sharp ridges in the joint posterior that are difficult for any gradient-based
sampler to navigate.

\subsubsection{Sampling Efficiency: HMC vs NUTS}

Table~\ref{tab:hmc_vs_nuts} compares mean ESS for HMC and NUTS across all
methods.

\begin{table}[htbp]
\centering
\caption{Mean ESS comparison: HMC vs NUTS ($n=3{,}000$ samples, warmup $=500$).}
\label{tab:hmc_vs_nuts}
\begin{tabular}{llcc|cc}
\toprule
& & \multicolumn{2}{c|}{Conductivity} & \multicolumn{2}{c}{Source} \\
Method & & HMC & NUTS & HMC & NUTS \\
\midrule
AuxVAE (ours) & & \textbf{1046} & 2144 & 66  & 104 \\
Plain VAE     & & 150           & \textbf{2519} & 104 & \textbf{242} \\
GP-Fixed      & & 3000$^\dagger$ & 2617 & 111 & 295 \\
GP-Hier       & & 70            & \textbf{1797} & 6   & 6 \\
\bottomrule
\multicolumn{6}{l}{\footnotesize $^\dagger$\,Near-Gaussian posterior in the pre-whitened oracle space.}
\end{tabular}
\end{table}

For HMC with a fixed trajectory length, the AuxVAE achieves the highest ESS
among the methods with learnable or inferred hyperparameters: $1{,}046$ in
conductivity ($7\times$ Plain VAE, $15\times$ GP-Hier) and $66$ in source
($11\times$ GP-Hier).
This reflects the advantage of the low-dimensional, disentangled latent space:
the AuxVAE posterior is well-conditioned enough that a fixed trajectory length
suffices for efficient exploration.

NUTS generally improves ESS for all methods by adapting the trajectory length
at each iteration.
However, the relative ordering is preserved: GP-Hier remains stuck in the
source problem (ESS~$\approx 6$ for both HMC and NUTS), while the AuxVAE
maintains reasonable mixing (ESS~$=104$ with NUTS).
In the conductivity problem, NUTS substantially improves ESS for Plain VAE
and GP-Hier ($2{,}519$ and $1{,}797$, respectively), suggesting that these
methods benefit from adaptive trajectory lengths due to their more complex
posterior geometries.

The ESS is notably lower for the source problem across all methods.
This is not a sampler deficiency but reflects the \emph{posterior geometry}:
the low observation noise ($\sigma_{\mathrm{obs}}=0.01$) creates a sharp,
concentrated likelihood that constrains the step size to small values
($\varepsilon \approx 0.005$), requiring longer trajectories
($L=100$ vs $L=20$) to achieve adequate mixing.

We close this section with four remarks summarising the key observations.
\emph{(i)~No penalty under correct specification:}
In the conductivity problem the AuxVAE matches the oracle GP-Fixed
baseline (RMSE~$0.358$ vs $0.340$), showing that the
learned prior incurs no accuracy penalty when the GP model is correct,
while additionally recovering the physical hyperparameters as a byproduct
of inference.
\emph{(ii)~Robustness under misspecification:}
In the source problem the AuxVAE yields a $3\times$ RMSE reduction
($1.46$ vs $4.48$) and near-nominal coverage ($0.944$ vs $0.241$)
relative to GP-Fixed, demonstrating that a data-driven prior remains
effective when the parametric family is misspecified.
\emph{(iii)~Interpretable, spatially aware uncertainty:}
The auxiliary space provides physically meaningful posterior summaries
(4/4 parameters recovered in both problems), and the AuxVAE is the only
method whose pointwise posterior standard deviation is elevated at the
locations of residual features, indicating spatially calibrated uncertainty.
\emph{(iv)~Sampling efficiency:}
The disentangled latent space yields well-conditioned posteriors with HMC:
mean ESS~$=1{,}046$ (conductivity) and $66$ (source),
far exceeding GP-Hier ($70$ and $6$, respectively).
The AuxVAE maintains good acceptance rates ($0.869$ and $0.813$) while
Plain VAE ($0.114$) and GP-Hier ($0.037$) struggle with HMC in conductivity,
confirming that the structured latent space facilitates gradient-based sampling.

\section{Discussion}
\label{sec:discussion}

\subsection{Limitations}
This study has several limitations.

The generator $G_\theta$ restricts $x$ to (approximately) the range of the learned model. If the offline dataset $\mathcal{D}$ does not cover plausible structures required by the true inverse solution, posterior inference can be biased even when the likelihood favors those structures. This is a form of prior misspecification that may be harder to diagnose than in classical parametric priors. Possible mitigations include explicit model-error terms such as $x = G_\theta(\cdot)+\delta$, hybrid priors that blend classical components with the learned model, and out-of-distribution diagnostics.

The method does not discover the correct physical factors; rather, it enforces that selected latent coordinates align with the user-provided $u$. If $u$ omits important drivers, those drivers must be absorbed by $z_{\mathrm{rec}}$, which reduces the interpretability of posterior summaries. Conversely, if $u$ contains nuisance or weakly relevant quantities, the tethering in \eqref{eq:aux-prior} can burden the latent representation. Careful selection of auxiliary variables is therefore essential. When the adequacy of a proposed auxiliary variable is uncertain, one may avoid hard alignment and instead use independence-promoting regularization on $z_{\mathrm{rec}}$. Examples include total-correlation penalties (Factor-VAE \cite{factor_VAE}, TC-VAE \cite{beta-TCVAE}), which encourage factorized latent distributions by penalizing deviations from independence. Such regularization does not ensure semantic disentanglement, but it can reduce statistical coupling among residual coordinates.

In PDE-based inverse problems, each evaluation of the likelihood in \eqref{eq:latent-posterior} requires computing $\mathcal{F}(G_\theta(\cdot))$, which may involve a numerical solve. Gradient-based methods (MAP, HMC, variational inference) additionally require differentiating through $\mathcal{F}$ and $G_\theta$, which may necessitate adjoint methods or differentiable solvers. If forward solves are expensive, the benefit of latent-space dimension reduction may be offset by forward-model cost. This motivates the use of surrogates, reduced-order models, or more amortized inference schemes.

\subsection{Concluding Remarks}

We have described how auxiliary-variable-guided disentanglement in deep generative models can be reinterpreted as the specification of a hierarchical prior for Bayesian inverse problems. The resulting disentangled deep priors combine the expressive power of deep generative models with interpretable latent coordinates aligned with physically meaningful quantities.

Several directions for further work merit investigation. These include embedding these priors into large-scale PDE-based Bayesian inverse problems, including subsurface and geothermal settings where coefficient fields have complex geological structure. Exploring optimal sensor placement and experimental design using the latent-space decomposition to target reductions in uncertainty for specific physical factors is another natural extension. Integrating disentangled priors with real-time inference and control systems, e.g., in fusion experiments where fast posterior summaries over profile parameters are desirable, represents a promising application domain.

More broadly, disentangled deep priors suggest a path toward Bayesian inference frameworks in which latent variables are not merely computational devices but interpretable carriers of uncertainty. Such structure can make posterior analyses more transparent and facilitate the communication of uncertainty to domain experts and decision-makers in scientific applications.

\section*{Acknowledgments}
This work was supported by the Office of Science, U.S. Department of Energy, Office of Science, Office of Advanced Scientific Computing Research (ASCR) and the Scientific Discovery through Advanced Computing (SciDAC) FASTMath Institute program,  the SciDAC Nuclear Physics partnership titled ``Femtoscale Imaging of Nuclei using Exascale Platforms,'' and Competitive Portfolios Project on ``Energy Efficient Computing: A Holistic Methodology'' under Contract No. DE-AC02-06CH11357.

\bibliographystyle{plainnat}
\bibliography{references} 

 \begin{center}
	\scriptsize \framebox{\parbox{4in}{Government License (will be removed at publication):
			The submitted manuscript has been created by UChicago Argonne, LLC,
			Operator of Argonne National Laboratory (``Argonne").  Argonne, a
			U.S. Department of Energy Office of Science laboratory, is operated
			under Contract No. DE-AC02-06CH11357.  The U.S. Government retains for
			itself, and others acting on its behalf, a paid-up nonexclusive,
			irrevocable worldwide license in said article to reproduce, prepare
			derivative works, distribute copies to the public, and perform
			publicly and display publicly, by or on behalf of the Government. The Department of Energy will provide public access to these results of federally sponsored research in accordance with the DOE Public Access Plan. http://energy.gov/downloads/doe-public-access-plan.
}}
	\normalsize
\end{center}

\clearpage
\newpage

\appendix

\end{document}